\documentclass[11pt,a4paper]{article}
\usepackage{amssymb}
\usepackage{amsmath}
\usepackage{amsfonts}
\usepackage{graphicx}
\usepackage{a4wide}
\usepackage{microtype}
\usepackage{bbm}
\usepackage{slashed}
\usepackage{amsthm}
\usepackage{mathrsfs}
\usepackage{hyperref}
\usepackage{color}

\definecolor{darkred}{rgb}{0.8,0.1,0.1}
\hypersetup{
     colorlinks=true,         % false: boxed links; true: colored links,false is default
     linkcolor=darkred,
     citecolor=blue,
}

\usepackage{comment}
\usepackage[margin=2.7cm]{geometry}
\usepackage[all,cmtip]{xy}

\usepackage{marvosym}

\theoremstyle{plain}
\newtheorem{theo}{Theorem}[section]
\newtheorem{lem}[theo]{Lemma}
\newtheorem{propo}[theo]{Proposition}
\newtheorem{cor}[theo]{Corollary}

\theoremstyle{definition}
\newtheorem{defi}[theo]{Definition}
\newtheorem{ex}[theo]{Example}

\newtheorem{rem}[theo]{Remark}

\numberwithin{equation}{section}

\def\Aff{\mathsf{Aff}}
\def\Vec{\mathsf{Vec}}
\def\AffBund{\mathsf{AffBund}}
\def\VecBund{\mathsf{VecBund}}
\def\Mod{\mathsf{Mod}}
\def\GlobHypAffGreen{\mathsf{GlobHypAffGreen}}
\def\VecBiLin{\mathsf{VecBiLin}}
\def\VecBiLinR{\mathsf{VecBiLinR}}

\def\PhaseSpace{\mathfrak{PhSp}}

\def\CCR{\mathfrak{CCR}}
\def\CAR{\mathfrak{CAR}}
\def\astAlg{{^\ast}\mathsf{Alg}}
\def\QFT{\mathfrak{A}}

\def\nn{\nonumber}

\def\bbR{\mathbb{R}}
\def\bbC{\mathbb{C}}
\def\bbN{\mathbb{N}}

\def\bfA{\mathsf{A}}
\def\bfV{\mathsf{V}}
\def\bfW{\mathsf{W}}

\def\MM{\mathcal{M}}
\def\AA{\mathcal{A}}

\def\EE{\mathcal{E}}

\def\Hom{\mathrm{Hom}}

\def\Ker{\mathrm{Ker}}

\def\id{\mathrm{id}}
\def\supp{\mathrm{supp}}

\def\vol{\mathrm{vol}_M}
\def\volo{\mathrm{vol}_{M_1}}
\def\volt{\mathrm{vol}_{M_2}}
\def\volthree{\mathrm{vol}_{M_3}}

\def\1{\mathbbm{1}}
\def\oone{\mathbf{1}}

\def\lin{\mathrm{lin}}

\def\DualConf{\mathfrak{DualConf}}

\newcommand{\ip}[2]{\langle #1,#2 \rangle}

\newcommand{\sect}[2]{\Gamma^\infty( #1 , #2 )}
\newcommand{\sectn}[2]{\Gamma_0^\infty( #1 , #2 )}
\newcommand{\sectsc}[2]{\Gamma_\mathrm{sc}^\infty( #1 , #2 )}

\def\sk{\vspace{2mm}}

\title{%
Quantum field theory on affine bundles
}

\author{%
Marco Benini$^{1,a}$, Claudio Dappiaggi$^{1,b}$ and Alexander Schenkel$^{2,c}$\vspace{4mm}\\
{\small $^1$ Dipartimento di Fisica}\\ 
{\small Universit{\`a} di Pavia \& INFN, sezione di Pavia –- Via Bassi 6, 27100 Pavia, Italia.}\vspace{2mm}\\
{\small $^2$ Fachgruppe Mathematik}\\
{\small Bergische~Universit\"at~Wuppertal,~Gau\ss stra\ss e~20,~42119~Wuppertal,~Germany.}\vspace{4mm}
\\
{\footnotesize $^a$ \texttt{marco.benini@pv.infn.it}~,~$^b$ \texttt{claudio.dappiaggi@unipv.it}~,~$^c$ \texttt{schenkel@math.uni-wuppertal.de} }
 }

\date{October 2012}

\begin{document}

\maketitle

\begin{abstract}
We develop a general framework for the quantization of bosonic and fermionic field theories
on affine bundles over arbitrary globally hyperbolic spacetimes. All concepts and results are formulated using the language
of category theory, which allows us to prove that these models satisfy the principle of
general local covariance. Our analysis is a preparatory step towards a full-fledged quantization scheme
for the Maxwell field, which emphasises the affine bundle structure of the 
bundle of principal $U(1)$-connections.
 As a by-product, our construction provides a new class of 
  exactly tractable locally covariant quantum field theories, which are a mild generalization of the linear ones.
  We also show the existence of a functorial 
assignment of linear quantum field theories to affine ones. The identification of suitable algebra homomorphisms enables us
to induce whole families of physical states (satisfying the microlocal spectrum condition) for affine quantum field theories
by pulling back quasi-free Hadamard states of the underlying linear theories.
\end{abstract}

%%%%%%%%%%%%%%%%%%%%%%%%%%%%%%%%%%%%%%%%%%%%%%%%%%%%%%%
%%%%%%%%%%%%%%%%%%%%%%%%%%%%%%%%%%%%%%%%%%%%%%%%%%%%%%%

\section{\label{sec:intro}Introduction}
In the past two decades the study of quantum field theories on curved spacetimes via algebraic
 techniques has witnessed considerable leap forwards, which have strengthened our 
 understanding of its mathematical and physical foundations. Particularly relevant has 
 been the formulation of the principle of general local covariance in \cite{Brunetti:2001dx}, 
 which characterizes a well-defined quantization scheme of any theory as a covariant functor 
 from the category of globally hyperbolic spacetimes to that of unital topological $\ast$-algebras. 
 In particular linear field theories, both bosonic and fermionic, have been thoroughly discussed within 
 this framework, see \cite{Bar:2007zz,Bar:2011iu}. Interactions have been dealt with only at a perturbative level 
via the extended algebra of Wick polynomials and the related time-ordered product \cite{Hollands:2001fb, Brunetti:2009qc}.

There are two weak points in our present understanding of locally covariant quantum field theories: 
The first is that all exactly constructable models so far are given by linear theories.
The second concerns local gauge theories and, most notably, the prime example of the Maxwell field. In a recent paper it has been shown that the standard quantization scheme applied to the field strength tensor, described by a dynamical $2$-form on 
the underlying spacetime, fails to satisfy the principle of general local covariance \cite{Dappiaggi:2011zs}.

The goal of this paper is to partly amend these deficiencies by considering affine field theories.
This name is prompted by the fact that the configuration bundle is an affine bundle  and the equation 
of motion operator is compatible with this affine structure. 
The motivation for developing these models is twofold: On the one hand, 
affine field theories are a mild generalization of linear field theories (in particular the 
equation of motion is not linear, however affine) and they can still be treated exactly within 
the standard scheme of locally covariant field theories. On the other hand, this paper is preparatory 
to a full-fledged analysis of the Maxwell field as a particular example of a
 Yang-Mills theory. From a mathematical point of view, the Maxwell field should be understood as a dynamical theory 
 of connections on principal $U(1)$-bundles. The basic geometric object in this description is
 the bundle of connections, which is an affine bundle modeled on the cotangent bundle of the spacetime. 
The Maxwell equation can be understood as a differential operator compatible with this affine structure.
However, since a key role in this case is played by the gauge group and by the 
gauge fixing procedure, we decided to devote a second paper to this specific model \cite{BDSMaxwell}. 
In this paper we shall provide a full-fledged characterization of the structural properties of affine field theories,
which are not subject to local gauge invariance. We name this type of model affine matter field theory. 
In order to achieve our goal, we exploit as much as possible the language of category theory and we think 
of our paper as a natural generalization of the work in \cite{Bar:2007zz,Bar:2011iu},
which addresses the structural properties of linear matter field theories.

We clarify the structure of the paper by outlining its content:
 In Section \ref{sec:notation1} we review the definition of an affine space and of its dual, showing in particular that 
it is a vector space. Affine bundles, their duals and the structure of their sets of sections are discussed
 in Section \ref{sec:notation2} and Section \ref{sec:notation3}.
 The remaining part of Section \ref{sec:notation} is devoted to summarize 
 well-known results on differential operators on vector bundles and Green-hyperbolic operators 
 on vector bundles over globally hyperbolic spacetimes. 
 In Section \ref{sec:affinediffop} we introduce the notion of affine differential operators and study in detail their formal adjoints. 
 Then in Section \ref{sec:classaff} we provide a characterization of affine matter field theories and prove important structural
 properties. Starting from Section \ref{sec:category} we show that the above
 concepts can be phrased in the language of category theory, which allows us in Section \ref{sec:quant} to formulate the quantization 
 of bosonic and fermionic affine matter field theories as a suitable covariant functor.
 The axioms of locally covariant quantum field theory are verified.
 In Section \ref{sec:linfunctor} we construct a functor associating linear quantum field theories to affine ones. 
 This is used in Section \ref{sec:states} to show that, starting from any quasi-free Hadamard state for the underlying linear theory, 
 it is possible to induce a family of states for the affine theory, which satisfies the microlocal spectrum condition.
 In Section \ref{sec:inhommat} we show that linear matter field theories coupled to external source terms 
 are naturally interpreted as affine matter field theories.
In order to be self-contained, we review in the Appendix \ref{app:CCRandCAR} 
 the construction of the CCR and CAR functors.

%%%%%%%%%%%%%%%%%%%%%%%%%%%%%%%%%%%%%%%%%%%%%%%%
%%%%%%%%%%%%%%%%%%%%%%%%%%%%%%%%%%%%%%%%%%%%%%%%

\section{\label{sec:notation}Preliminaries and notation}
In this work we focus for the sake of simplicity on real field theories. 
Complex field theories follow accordingly with minor modifications.
Thus, all vector spaces will be taken over the field $\bbR$
if not stated otherwise.
\subsection{\label{sec:notation1}Affine spaces, duals and morphisms}
We briefly review the basic concepts of affine spaces. All vector spaces in this section will be finite-dimensional.
\begin{defi}\label{affinespace}
\begin{itemize}
\item[(i)]An {\bf affine space} is a triple $(A,V,\Phi)$ consisting of a set $A$, a vector space $V$ 
and an action $\Phi: A\times V\to A$ of the abelian group $(V,+)$ on $A$, which is free and transitive.
We frequently use the convenient notation $\Phi(a,v) = a+v$ in which the group action properties
simply read $a+0 =a$ and $(a+v)+w = a+(v+w) =: a+v+w$, for all $a\in A$ and $v,w\in V$.
For $a,a^\prime\in A$ we denote by $a^\prime-a:= v$ the unique element $v\in V$ such that
$a^\prime = a +v$.
If it does not cause confusion we simply write $A$ for the affine space $(A,V,\Phi)$.
\item[(ii)] Given two affine spaces $(A_1,V_1,\Phi_1)$ and $(A_2,V_2,\Phi_2)$,
we call a map $f: A_1\to A_2$ an {\bf affine map}, if 
there exists a linear map $f_V:V_1\to V_2$, such that for all $a\in A_1$ and $v\in V_1$,
\begin{flalign}
f\big(\Phi_1(a,v)\big) = \Phi_2\big(f(a),f_V(v)\big)~,
\end{flalign}
or in the more convenient notation $f(a+v) = f(a)+f_V(v)$, where we used for notational
simplicity the same symbol ``+'' for both, $\Phi_1$ and $\Phi_2$.
We call the map $f_V$ the {\bf linear part} of $f$.
\end{itemize}
\end{defi}
\sk
\begin{rem}
The linear part of an affine map $f:A_1\to A_2$ is unique:
Let us assume that there are two linear parts $f_V,\widetilde{f}_V:V_1\to V_2$. Then by definition
we have, for all $a\in A_1$ and $v\in V_1$, $f(a+v) = f(a) +f_V(v) = f(a) + \widetilde{f}_V(v)$.
This implies $f_V(v) = \widetilde{f}_V(v)$, for all $v\in V_1$, and hence $f_V =\widetilde{f}_V$.
\end{rem}
\sk

\begin{ex}
Every vector space $V$ can be regarded as an affine space modeled on itself: Choose $A=V$ (as sets) and define
$\Phi$ via the abelian group structure $+$ on $V$.
In this case the affine endomorphisms are given by affine transformations, 
i.e.~maps $f:A\to A\,,~v\mapsto f(v) = f_V(v) +b$, where $f_V:V\to V$ is a linear map and $b\in V$. 
The linear part of this map is $f_V$.
\end{ex}
\begin{ex}\label{ex:seq}
Consider a short exact sequence of vector spaces and linear maps
\begin{flalign}
\xymatrix{
0\ar[r] &  W_1 \ar[r]^-{f}& W_2 \ar[r]^-{g} & W_3 \ar[r] & 0~.
}
\end{flalign}
We say that a linear map $a: W_3\to W_2$ is a splitting of this sequence, if
$g\circ a =\id_{W_3}$.
Let us denote by $A := \big\{a \in \Hom_\bbR(W_3,W_2) : g\circ a =\id_{W_3}\big\}\subseteq \Hom_\bbR(W_3,W_2)$
  the set of all splittings and
by $V:= \Hom_\bbR(W_3,W_1)$ the vector space of linear maps from $W_3$ to $W_1$.
Consider the following map
\begin{flalign}
\Phi: A\times V \to A ~,~~(a,v) \mapsto \Phi(a,v) = a + f\circ v~.
\end{flalign}
The map $\Phi(a,v):W_3\to W_2$  is linear and satisfies $g\circ \Phi(a,v)=\id_{W_3}$, for
all $a\in A$ and $v\in V$. Furthermore, $\Phi$ defines an action of the abelian group $(V,+)$ on $A$, 
i.e.~for all $a\in A$, $\Phi(a,0) = a$ and, for all $a\in A$ and $v,w\in V$,  $\Phi(\Phi(a,v),w) = 
a+f\circ v + f\circ w = a + f\circ (v+w) = \Phi(a,v+w)$. Since $f$ is injective, the group action $\Phi$
is free. It is also transitive: Let $a,a^\prime\in A$ be arbitrary, then $a^\prime-a:W_3\to \Ker(g)\subseteq W_2$ is a linear map
with values in the kernel of the map $g:W_2\to W_3$. Since the sequence is exact, we have that $\mathrm{Im}(f) = \Ker(g)$
and that $f:W_1\to \mathrm{Im}(f)\subseteq W_2$ is a linear isomorphism. Defining
$v:= f^{-1}\circ(a^\prime-a)$ we obtain $\Phi(a,v) = a + a^\prime-a = a^\prime$. Hence, the triple
$(A,V,\Phi)$ is an affine space.
\end{ex}
\sk
\begin{rem}\label{rem:affinevecdiffeo}
Let $(A,V,\Phi)$ be an affine space.
Since $V$ is a finite-dimensional vector space, it comes with a  canonical topology induced from an Euclidean norm
and with a canonical $C^\infty$-structure induced by $\bbR^n$.
Fixing any element $a\in A$ we obtain a bijection of sets $\Phi_a:=\Phi(a,\,\cdot\,) : V \to A$. We induce a
topology on $A$ by $ \big\{U\subseteq A: \Phi_a^{-1}[U] \subseteq V\text{ is open}\big\}$.
This induced topology does not depend on the choice of $a\in A$. Even more, $A$ is a $C^\infty$-manifold:
The inverse of the map $\Phi_a : V \to A$ provides for any choice of $a\in A$ a global coordinate chart.
The transition function for two charts specified by $a,a^\prime\in A$, i.e.~the map 
$\Phi_{a^\prime}^{-1}\circ \Phi_a :V\to V$, is a polynomial of degree one and thus smooth.
Notice that the maps $\Phi : A\times V\to A\,,~(a,v)\mapsto a+v$ and $A\times A\to V\,,~(a,a^\prime)\mapsto a-a^\prime$
are smooth. Furthermore, every affine map $f:A_1\to A_2$ is smooth.
\end{rem}
\sk

Of particular importance for us is the dual  of an affine space.
\begin{defi}
Let $(A,V,\Phi)$ be an affine space. The {\bf vector dual} $A^\dagger$ of $(A,V,\Phi)$ is defined to be
the set of all affine maps $\varphi: A\to\bbR$.
\end{defi}
\sk

We review the following standard lemma:
\begin{lem}\label{lem:vectordual}
The vector dual $A^\dagger$ of an affine space $(A,V,\Phi)$ is a vector space
 of dimension $\mathrm{dim}(A^\dagger) = \mathrm{dim}(V) + 1$.
\end{lem}
\begin{proof}
The fact that $A^\dagger$ is a vector space is standard, since the target space of the affine maps $\varphi:A\to \bbR$
is a vector space.

To obtain the dimension of $A^\dagger$ we construct a basis. The constant map
$\1:A\to\bbR\,,~a\mapsto \1(a)=1$ is an element in $A^\dagger$ with linear part $\1_V =0$. 
We fix an arbitrary element $\widehat{a}\in A$ and define for each element $v^\ast\in V^\ast$,
with $V^\ast$ denoting the dual vector space of $V$, an affine map 
$\varphi^{v^\ast}:A\to\bbR\,,~a\mapsto \varphi^{v^\ast} (a) = v^\ast(a-\widehat{a})$.
The linear part of $\varphi^{v^\ast}$ is $v^\ast$. A basis of $V^\ast$ together with $\1$ provides a basis of $A^\dagger$,
thus $\mathrm{dim}(A^\dagger) = \mathrm{dim}(V^\ast)+1 = \mathrm{dim}(V)+1$.
\end{proof}
\sk

Let $(A_i,V_i,\Phi_i)$, $i=1,2,3$, be affine spaces. The usual composition $\circ$ of two affine maps
$f:A_1\to A_2$ and $g: A_2\to A_3$ is again an affine map
$g\circ f:A_1\to A_3$ with linear part
$ (g\circ f)_V = g_V\circ f_{V}$.
This allows us to define a category of affine spaces. We restrict the morphisms in this category to affine isomorphisms,
 since this is the structure we require later. 
\begin{defi}
The category $\Aff$ consists of the following objects and morphisms:
\begin{itemize}
\item An object in $\Aff$ is an affine space $(A,V,\Phi)$.
\item A morphism between two objects $(A_1,V_1,\Phi_1)$ and $(A_2,V_2,\Phi_2)$ in $\Aff$ is an affine isomorphism
$f:A_1\to A_2$.
\end{itemize} 
The category $\Vec$ consists of the following objects and morphisms:
\begin{itemize}
\item An object in $\Vec$ is a vector space $V$.
\item A morphism between two objects $V_1$ and $V_2$ in $\Vec$ is a linear isomorphism $f:V_1\to V_2$.
\end{itemize}
\end{defi}
\sk

The association of the underlying vector space $V$ to an affine space $(A,V,\Phi)$ is functorial:
\begin{lem}
The mapping $\mathfrak{Lin}:\Aff\to \Vec$  specified on objects by
$\mathfrak{Lin}(A,V,\Phi) =V$ and on morphisms by $\mathfrak{Lin}(f) = f_V$ is a covariant functor.
\end{lem}
\begin{proof}
For every object $(A,V,\Phi)$ in $\Aff$, $\mathfrak{Lin}(A,V,\Phi)=V$ is by definition a vector space.
For every morphism $f:A_1\to A_2$ in $\Aff$ the linear part $\mathfrak{Lin}(f)=f_V:V_1\to V_2$ is a linear isomorphism,
with inverse ${f_V}^{-1}={f^{-1}}_V:V_2\to V_1$ given by the linear part of the inverse affine map $f^{-1}:A_2\to A_1$.
We also have that $\mathfrak{Lin}(\id_A) = \id_V$ and $\mathfrak{Lin}(g\circ f) = (g\circ f)_V= g_V\circ f_V
= \mathfrak{Lin}(g)\circ\mathfrak{Lin}(f)$, for every composable morphisms $f,g$ in $\Aff$.
\end{proof}
The association of the vector dual $A^\dagger$ to an affine space $(A,V,\Phi)$ is also functorial:
\begin{lem}\label{lem:dualfunctor}
The mapping $\mathfrak{Dual}:\Aff\to \Vec$ specified on objects by 
$\mathfrak{Dual}(A,V,\Phi)=A^\dagger$ and on morphisms by $\mathfrak{Dual}(f) = f^\dagger$, with
\begin{flalign}\label{eqn:dualmap}
f^\dagger :A_1^\dagger \to A_2^\dagger~,~~\varphi \mapsto f^\dagger(\varphi) = \varphi\circ f^{-1}~,
\end{flalign}
is a covariant functor.
\end{lem}
\begin{proof}
For every object $(A,V,\Phi)$ in $\Aff$, $\mathfrak{Dual}(A,V,\Phi)=A^\dagger$ is a vector space due to
 Lemma \ref{lem:vectordual}. For every morphism $f:A_1\to A_2$ in $\Aff$, the dual map
 $\mathfrak{Dual}(f)=f^\dagger :A_1^\dagger \to A_2^\dagger$ given in (\ref{eqn:dualmap}) is a linear isomorphism,
 with inverse given by ${f^\dagger}^{-1}={f^{-1}}^\dagger=\mathfrak{Dual}(f^{-1})$,
  where $f^{-1}:A_2\to A_1$ is the inverse affine map.
 We also have that $\mathfrak{Dual}(\id_A) = \id_{A^\dagger}$ and 
 $\mathfrak{Dual}(g\circ f) = \mathfrak{Dual}(g)\circ\mathfrak{Dual}(f)$, for every
composable morphisms $f,g$ in $\Aff$.
\end{proof}

%%%%%%%%%%%%%%%%%%%%%%%%%%%%%%%%%%%%%%%%%%%%%%%%

\subsection{\label{sec:notation2}Affine bundles, duals and morphisms}
All manifolds will be  of class $C^\infty$, Hausdorff and second countable.
All maps between manifolds are $C^\infty$.
\sk

A {\bf fibre bundle} is a quadruple $(E,M,\pi,F)$, where $E$, $M$, $F$ are manifolds and $\pi:E\to M$ is 
a surjection. 
We require  that the fibre bundle is locally trivializable,
i.e.~that for every $x\in M$ there exists an open neighbourhood $U\subseteq M$ and 
a diffeomorphism $\psi:\pi^{-1}[U] \to U\times F$, such that the following diagram commutes
\begin{flalign}
\xymatrix{
\pi^{-1}[U] \ar[d]_-{\pi} \ar[rr]^-{\psi} & &U\times F\ar[dll]^-{\mathrm{pr}_1}\\
U &&
}
\end{flalign}
where $\mathrm{pr}_1$ denotes the canonical projection on the first factor.
We call $E$ the total space, $M$ the base space, $\pi$ the projection,
$F$ the typical fibre and $E\vert_x := \pi^{-1}[\{x\}]$ the fibre over $x\in M$.
 We also call the pair $(U,\psi)$ a local bundle chart.
 Notice that a local bundle chart $(U,\psi)$  provides
 a diffeomorphism, for all $y\in U$, $\psi\vert_y : E\vert_y \to F$.
 If it is convenient we also denote the fibre bundle $(E,M,\pi,F)$ simply by $E$.
 
Given two fibre bundles $(E_i,M_i,\pi_i,F_i)$, $i=1,2$, a {\bf fibre bundle map} is a pair of maps
$(f:E_1\to E_2,\underline{f}:M_1\to M_2)$, such that the following diagram commutes
\begin{flalign}
\xymatrix{
E_1\ar[d]_-{\pi_1}\ar[rr]^-{f} & & E_2\ar[d]^-{\pi_2}\\
M_1\ar[rr]^-{\underline{f}} && M_2
}
\end{flalign}
If the source and target of the fibre bundle map are clear, we also simply write
$(f,\underline{f})$
\sk

A {\bf vector bundle} is a fibre bundle $(\bfV,M,\rho,V)$, such that all fibres $\bfV\vert_x =\rho^{-1}[\{x\}]$, $x\in M$,
and the typical fibre $V$ are vector spaces, and such that, for all $x\in M$, there exists a local bundle chart $(U,\psi)$ with
 $\psi\vert_y:\bfV\vert_y \to V$ being a linear isomorphism, for all $y\in U$.
A local bundle chart of this kind is called a local vector bundle chart. 
If it is convenient we also denote the vector bundle $(\bfV,M,\rho,V)$ simply by $\bfV$.

Given two vector bundles $(\bfV_i,M_i,\rho_i,V_i)$, $i=1,2$, a {\bf vector bundle map} is a fibre bundle map
$(f:\bfV_1\to \bfV_2,\underline{f}:M_1\to M_2)$, such that, for all $x\in M_1$,
$f\vert_{x}: \bfV_1\vert_x\to \bfV_2\vert_{\underline{f}(x)}$  is a linear map. If the source and target are clear,
we also simply write $(f,\underline{f})$ for a vector bundle map
\sk

We define the concept of affine bundles following \cite[Chapter 6.22]{KMS}.
\begin{defi}\label{def:affinebundle}
\begin{itemize}
\item[(i)] An {\bf affine bundle} is a triple $\big(M,(\bfA,M,\pi,A),(\bfV,M,\rho,V)\big)$, where $M$ is a manifold,
$(\bfA,M,\pi,A)$ is a fibre bundle over $M$ and $(\bfV,M,\rho,V)$ is a vector bundle over $M$, such that
\begin{itemize}
\item for all $x\in M$, the fibre $\bfA\vert_x$ is an affine space modeled on $\bfV\vert_x$,
\item the typical fibre $A$ is an affine space modeled on the typical fibre $V$,
\item for all $x\in M$, there exists a local  bundle chart $(U,\psi)$ of $(\bfA,M,\pi,A)$ and a local vector bundle
chart $(U,\psi_\bfV)$ of $(\bfV,M,\rho,V)$,
such that, for all $y\in U$,
$\psi\vert_y:\bfA\vert_y\to A$ is an affine isomorphism
with linear part ${\psi\vert_y}_V =\psi_\bfV\vert_y:\bfV\vert_y \to V$. 
We call the triple $(U,\psi,\psi_\bfV)$ a local affine bundle chart. 
\end{itemize}
If it does not cause confusion we simply write $(M,\bfA,\bfV)$ or even only $\bfA$ for an affine bundle.
\item[(ii)] Given two affine bundles $\big(M_i,(\bfA_i,M_i,\pi_i,A_i),(\bfV_i,M_i,\rho_i,V_i)\big)$, $i=1,2$,
an {\bf affine bundle map} is a fibre bundle map $(f:\bfA_1\to \bfA_2,\underline{f}:M_1\to M_2)$, such that,
for all $x\in M_1$, $f\vert_x:\bfA_1\vert_x\to \bfA_2\vert_{\underline{f}(x)}$ is an affine map.

If the source and target are clear, we also write for the affine bundle map simply $(f,\underline{f})$.
\end{itemize}
\end{defi}
\sk

\begin{rem}\label{rem:uniquelinmap}
Every affine bundle map $(f,\underline{f})$ induces a unique vector bundle map $(f_\bfV,\underline{f})$ on the underlying
vector bundles. Indeed, the map $f_\bfV:\bfV_1\to \bfV_2$ is uniquely specified by, for all $x\in M_1$,
$f_\bfV\vert_x :={ f\vert_x}_V : \bfV_1\vert_x\to\bfV_2\vert_{\underline{f}(x)}$, where
${ f\vert_x}_V$ denotes the linear part of the affine map  $f\vert_x:\bfA_1\vert_x\to \bfA_2\vert_{\underline{f}(x)}$.
We call $(f_\bfV,\underline{f})$ the {\bf linear part} of the affine bundle map $(f,\underline{f})$.
\end{rem}
\sk

\begin{ex}\label{ex:vectorbundasaffinebund}
Every vector bundle $(\bfV,M,\rho,V)$ can be regarded as an affine bundle modeled on itself:
Simply choose $(\bfA,M,\pi,A)= (\bfV,M,\rho,V)$ (as fibre bundles). The three conditions in Definition
 \ref{def:affinebundle} (i)  are easily verified.
Let $(f_\bfV,\underline{f})$ be a vector bundle endomorphism and $b:M\to\bfV$ be a section.
Then defining
\begin{flalign}
f:\bfA\to\bfA~,~~v\mapsto f(v) = f_\bfV(v) + b\big(\underline{f}(\rho(v))\big)~
\end{flalign}
gives us an affine bundle endomorphism $(f,\underline{f})$ with linear part $(f_\bfV,\underline{f})$.
\end{ex}
\begin{ex}\label{ex:sequenceaffinebund}
Consider a short exact sequence of vector bundles over $M$ and vector bundle maps (covering $\id_M$)
\begin{flalign}\label{eqn:vecbundsequence}
\xymatrix{
M\times 0 \ar[r] & \bfW_1 \ar[r]^-{f} & \bfW_2 \ar[r]^-{g} & \bfW_3 \ar[r] & M\times 0~.
}
\end{flalign}
We say that a vector bundle map $(a:W_3\to W_2,\id_M:M\to M)$ is a splitting  of this sequence, if
$g\circ a =\id_{\bfW_3}$. In order to characterize these splittings we introduce the homomorphism bundles, 
for all $i,j\in \{1,2,3\}$, $\big(\Hom(\bfW_i,\bfW_j),M,\rho_{ij},\Hom_\bbR(W_i,W_j)\big)$.
Notice that the maps $f$ and $g$ in (\ref{eqn:vecbundsequence}) can be identified with sections
of $\Hom(\bfW_1,\bfW_2)$ and $\Hom(\bfW_2,\bfW_3)$, respectively. Similarly, the splitting $a$ 
can be identified with a section of a subbundle of $\Hom(\bfW_3,\bfW_2)$, which we want to study now.
We define the submanifold $\bfA:= \big\{a\in \Hom(\bfW_3,\bfW_2) : g\circ a = \id_{\bfW_3\vert_{\rho_{32}(a)}}\big\}$
and consider the induced subbundle $\big(\bfA,M,\pi,A\big)\subseteq
 \big(\Hom(\bfW_3,\bfW_2),M,\rho_{32},\Hom_\bbR(W_3,W_2)\big)$.
The fibre bundle $\bfA$ is an affine bundle over $\big(\Hom(\bfW_3,\bfW_1),M,\rho_{31},\Hom_\bbR(W_3,W_1)\big)$.
This can be shown by using local vector bundle charts of the homomorphism bundles
and arguments as in Example \ref{ex:seq}.

A particular example of a short exact sequence of vector bundles (\ref{eqn:vecbundsequence}) is the Atiyah
sequence \cite{Atiyah,Lopez} associated to a principal bundle. Splittings of this sequence
are in one-to-one correspondence with principal connections, hence, the affine bundle
$\big(M,\bfA,\bfV\big)$ associated to the Atiyah sequence is of utmost importance in the study of
gauge theories. For a study of the Maxwell field on curved spacetimes using this affine bundle formulation
we refer to our forthcoming work \cite{BDSMaxwell}.
\end{ex}
\sk

The vector dual bundle is characterized by a standard construction:
\begin{defi}
Let  $\big(M,(\bfA,M,\pi,A),(\bfV,M,\rho,V)\big)$ be an affine bundle. The {\bf vector dual bundle} is the vector bundle
specified by $(\bfA^\dagger,M,\pi^\dagger,A^\dagger)$, where $A^\dagger$ is the vector dual of $A$, 
$\bfA^\dagger :=\bigcup_{x\in M} \bfA\vert_x^\dagger$ and $\pi^\dagger:\bfA^\dagger \to M$
is defined by, for all $\varphi\in\bfA\vert_x^\dagger$ and $x\in M$, $\pi^\dagger(\varphi) = x$.
If it is convenient we also write simply $\bfA^\dagger$ for the vector dual bundle.
\end{defi}

\sk

Let $\big(M_i,\bfA_i,\bfV_i\big)$, $i=1,2,3$, be affine bundles,
$(f,\underline{f})$ an affine bundle map between $\big(M_1,\bfA_1,\bfV_1\big)$
and $\big(M_2,\bfA_2,\bfV_2\big)$ and $(g,\underline{g})$ an affine bundle map between
$\big(M_2,\bfA_2,\bfV_2\big)$ and $\big(M_3,\bfA_3,\bfV_3\big)$. 
The composition $(g,\underline{g})\circ (f,\underline{f}):=
(g\circ f,\underline{g}\circ\underline{f})$ is again an affine bundle map. The linear part is
 $((g\circ f)_\bfV,\underline{g}\circ\underline{f})=
 (g_\bfV\circ f_\bfV,\underline{g}\circ\underline{f})= (g_\bfV,\underline{g})\circ(f_\bfV,\underline{f})$.
This allows us to define a category of affine bundles. We use a restricted class of morphisms, since this
is the structure we later require.
\begin{defi}
The category $\AffBund$ consists of the following objects and morphisms: 
\begin{itemize}
\item An object in $\AffBund$ is an affine bundle  $\big(M,(\bfA,M,\pi,A),(\bfV,M,\rho,V)\big)$.
\item A morphism between two objects  $\big(M_i,(\bfA_i,M_i,\pi_i,A_i),(\bfV_i,M_i,\rho_i,V_i)\big)$, $i=1,2$,  in $\AffBund$ 
is an affine bundle map $(f:\bfA_1\to \bfA_2,\underline{f}:M_1\to M_2)$, such that
$\underline{f}$ is an embedding, $\underline{f}[M_1] \subseteq M_2$ is open,
 and, for all $x\in M_1$, $f\vert_x:\bfA_1\vert_x\to\bfA_2\vert_{\underline{f}(x)}$ is an affine isomorphism.
\end{itemize}
The category $\VecBund$ consists of the following objects and morphisms:
\begin{itemize}
\item An object in $\VecBund$ is a vector bundle $(\bfV,M,\rho,V)$.
\item A morphism between two objects $(\bfV_i,M_i,\rho_i,V_i)$, $i=1,2$, in $\VecBund$
is a vector bundle map $(f:\bfV_1\to\bfV_2,\underline{f}:M_1\to M_2)$, such that
$\underline{f}$ is an embedding, $\underline{f}[M_1] \subseteq M_2$ is open,
 and, for all $x\in M_1$, $f\vert_x:\bfV_1\vert_x\to\bfV_2\vert_{\underline{f}(x)}$ is a linear isomorphism.
\end{itemize}
\end{defi}
\sk

\begin{rem}\label{rem:vecinv}
Let $(f,\underline{f})$ be a morphism in $\AffBund$.
From the property that, for all $x\in M_1$, $f\vert_x:\bfA_1\vert_x\to\bfA_2\vert_{\underline{f}(x)}$ is an affine isomorphism
it follows that, for all $x\in M_1$, $f_\bfV\vert_x={f\vert_x}_V:\bfV_1\vert_x\to\bfV_2\vert_{\underline{f}(x)}$
 is a linear isomorphism.
\end{rem}
\sk

The association of the underlying vector bundle is functorial:
\begin{lem}\label{lem:linearpartfunctor}
There is a covariant functor $\mathfrak{LinBund}:\AffBund\to\VecBund$.  It is specified on objects by 
$\mathfrak{LinBund}\big(M,\bfA,\bfV\big) =\bfV$ and on morphisms
by $\mathfrak{LinBund}(f,\underline{f}) = (f_\bfV,\underline{f})$.
\end{lem}
\begin{proof}
Follows from the composition property of affine bundle maps and Remark \ref{rem:vecinv}.
\end{proof}

Also the association of the vector dual bundle is functorial:
\begin{lem}\label{lem:dualbundfunctor}
There is a covariant functor $\mathfrak{DualBund}:\AffBund\to\VecBund$. It is specified on objects by 
$\mathfrak{DualBund}\big(M,\bfA,\bfV\big) =\bfA^\dagger$ and on morphisms
by $\mathfrak{DualBund}(f,\underline{f}) = (f^\dagger:\bfA_1^\dagger\to\bfA_2^\dagger,\underline{f}:M_1\to M_2)$,
 which is the vector bundle map defined by, for all $x\in M_1$, 
 \begin{flalign}
 f^\dagger\vert_x: \bfA_1^\dagger\vert_x \to \bfA_2^\dagger\vert_{\underline{f}(x)}~,~~\varphi \mapsto f^\dagger\vert_x(\varphi)=\varphi\circ f\vert_x^{-1}~.
 \end{flalign}
\end{lem}
\begin{proof}
All functor properties are shown analogously to Lemma \ref{lem:dualfunctor}.
\end{proof}

%%%%%%%%%%%%%%%%%%%%%%%%%%%%%%%%%%%%%%%%%%%

\subsection{\label{sec:notation3}Sections of affine bundles}
Before studying sections of affine bundles we remind the reader of the following construction:
Given two fibre bundles $(E_1,M,\pi_1,F_1)$ and $(E_2,M,\pi_2,F_2)$ over the same 
base space $M$, we can consider their fibred product bundle $(E_1\times_ME_2,M,\Pi,F_1\times F_2)$.
The total space of this fibre bundle is defined by
 $E_1\times_M E_2 := \big\{(e_1,e_2)\in E_1\times E_2 : \pi_1(e_1) = \pi_2(e_2)\big\}\subset E_1\times E_2$
 and the projection by $\Pi:E_1\times_ME_2 \to M\,,~(e_1,e_2)\mapsto \Pi(e_1,e_2) = \pi_1(e_1) = \pi_2(e_2)$.
 This fibre bundle is locally trivializable as follows: Take any $x\in M$ and local bundle charts  $(U_i,\psi_i)$
  of $(E_i,M,\pi_i,F_i)$, $i=1,2$. We obtain a local bundle chart $(U_1\cap U_2,\psi_1\times_M\psi_2)$ of
  $(E_1\times_ME_2,M,\Pi,F_1\times F_2)$ by setting
\begin{flalign}
\nn\psi_1\times_M\psi_2 :\Pi^{-1}[U_1\cap U_2] &\to U_1\cap U_2\times F_1\times F_2~,~~\\
(e_1,e_2)&\mapsto\psi_1\times_M\psi_2(e_1,e_2)= \big(\mathrm{pr}_{1}(\psi_1(e_1)),\mathrm{pr}_2(\psi_1(e_1)),\mathrm{pr}_2(\psi_2(e_2))\big)~.
\end{flalign}
Let $(\bfV,M,\rho,V)$ be a vector bundle. Due to the abelian group structure on the fibres $\bfV\vert_x$, for all
$x\in M$, there is a canonical fibre bundle map  (covering $\id_M$)
\begin{flalign}
+:\bfV\times_M\bfV\to \bfV ~,~~(v,w)\mapsto v+w~.
\end{flalign}
Let now $\big(M,(\bfA,M,\pi,A),(\bfV,M,\rho,V)\big)$ be an affine bundle. Similar to above, we have 
the canonical fibre bundle maps (covering $\id_M$)
\begin{subequations}
\begin{flalign}
-:\bfA\times_M\bfA\to \bfV ~,~~(a,b)\mapsto a-b~,
\end{flalign}
and
\begin{flalign}
\Phi: \bfA\times_M\bfV \to \bfA~,~~(a,v)\mapsto a+v~.
\end{flalign}
\end{subequations}
The following diagram of fibre bundle maps (covering $\id_M$) commutes
\begin{flalign}\label{eqn:affbundopcd}
\xymatrix{
\ar[d]_-{\id_\bfA \times +}\bfA\times_M\bfV\times_M\bfV \ar[rr]^-{\Phi\times \id_\bfV} & & \ar[d]^-{\Phi}\bfA\times_M\bfV\\
\bfA\times_M\bfV \ar[rr]^-{\Phi} & & \bfA
}
\end{flalign}
Let $(\bfA^\dagger,M,\pi^\dagger,A^\dagger)$ be the vector dual bundle. There is a canonical
fibre  bundle map (covering $\id_M$)
\begin{flalign}\label{eqn:evmap}
\mathrm{ev}: \bfA^\dagger\times_M\bfA \to M\times \bbR~,~~(\varphi,a)\mapsto (\pi(a),\varphi(a))~.
\end{flalign}
We shall also denote the evaluation of $\varphi$ on $a$ simply by $\varphi(a)$.
\sk

Let  $\big(M,(\bfA,M,\pi,A),(\bfV,M,\rho,V)\big)$ be an affine bundle.
We denote by $\sect{M}{\bfA}$ the set of sections of the fibre bundle $(\bfA,M,\pi,A)$
and by $\sect{M}{\bfV}$ the $C^\infty(M)$-module of sections of
the underlying vector bundle $(\bfV,M,\rho,V)$. 
As a consequence of Remark \ref{rem:affinevecdiffeo} and \cite[Chapter I, Theorem 5.7]{Kobayashi},
the set of sections $\sect{M}{\bfA}$ is never empty.
We define the map $\Phi_\Gamma: \sect{M}{\bfA}\times \sect{M}{\bfV} \to\sect{M}{\bfA}$
by, for all $s\in\sect{M}{\bfA}$, $\sigma\in\sect{M}{\bfV}$ and $x\in M$, $\Phi_\Gamma(s,\sigma)(x) :=
(s + \sigma)(x) := s(x)+\sigma(x)$.
Notice that the following diagram commutes
\begin{flalign}
\xymatrix{
\bfA\times_M\bfV \ar[rr]^-{\Phi} & & \bfA\\
M \ar[u]^-{s\times_M\sigma} \ar[urr]_-{\qquad \Phi_\Gamma(s,\sigma)=:s+\sigma}& & 
}
\end{flalign}
where $s\times_M\sigma$ is the section specified by, for all $x\in M$, $(s\times_M\sigma)(x) := \big(s(x),\sigma(x)\big)$.
\begin{lem}
The triple
$\big(\sect{M}{\bfA},\sect{M}{\bfV},\Phi_\Gamma\big)$ is an (infinite-dimensional) affine space.
\end{lem}
\begin{proof}
Clearly $\sect{M}{\bfA}$ is a set and $\sect{M}{\bfV}$ a vector space. The map $\Phi_\Gamma$ is an action
of the abelian group  $\big(\sect{M}{\bfV},+\big)$ on $\sect{M}{\bfA}$. We show that this action is transitive:
Let $s,s^\prime\in\sect{M}{\bfA}$ be arbitrary. We define an element $\sigma\in\sect{M}{\bfV}$ via the following commutative
diagram
\begin{flalign}
\xymatrix{
\bfA\times_M\bfA \ar[rr]^-{-} & &\bfV\\
M\ar[u]^-{s^\prime\times_M s} \ar[urr]_-{\sigma} &&
}
\end{flalign}
That is, for all $x\in M$, $\sigma(x) = s^\prime(x) -s(x)$. Hence, for all $x\in M$,
$\Phi_\Gamma(s,\sigma)(x) = s(x) +\sigma(x) = s^\prime(x)$.
The action is free, since given any $s\in \sect{M}{\bfA}$ and $\sigma\in \sect{M}{\bfV}$,
the condition $s = \Phi_\Gamma(s,\sigma)$ implies that $\sigma =0$.
\end{proof}
\sk

We can also consider the $C^\infty(M)$-module of sections $\sect{M}{\bfA^\dagger}$ of the vector dual bundle.
The canonical evaluation map
(\ref{eqn:evmap})  induces a point-wise evaluation map
on sections (denoted with a slight abuse of notation by the same symbol)
 \begin{flalign}
\mathrm{ev}:\sect{M}{\bfA^\dagger}\times \sect{M}{\bfA}\to C^\infty(M)~,~~(\varphi,s)\mapsto \varphi(s)~.
\end{flalign} Furthermore, due to Remark \ref{rem:uniquelinmap}, there exists
 for every $\varphi\in\sect{M}{\bfA^\dagger}$ a linear part $\varphi_{\bfV^\ast} \in \sect{M}{\bfV^\ast}$, with 
 $\bfV^\ast$ denoting the dual vector bundle of $\bfV$, such that for all $s\in \sect{M}{\bfA}$ and $\sigma\in \sect{M}{\bfV}$,
 \begin{flalign}\label{eqn:sectionaff}
 \varphi(s+\sigma) = \varphi(s)+ \varphi_{\bfV^\ast}(\sigma)~.
 \end{flalign}
 If we assume that the vector bundle $\bfV$ comes with a nondegenerate bilinear
 form $\ip{~}{~}$, the dual vector bundle $\bfV^\ast$ can be canonically identified with the vector bundle $\bfV$. 
 Correspondingly, we can also identify $\sect{M}{\bfV^\ast}$  with $\sect{M}{\bfV}$ and 
 for every $\varphi\in\sect{M}{\bfA^\dagger}$ there exists a $\varphi_\bfV\in \sect{M}{\bfV}$ (also called the linear part), such that
  for all $s\in \sect{M}{\bfA}$ and $\sigma\in \sect{M}{\bfV}$,
  \begin{flalign}
  \varphi(s+\sigma) = \varphi(s) + \ip{\varphi_\bfV}{\sigma}~.
  \end{flalign}

 \begin{rem}\label{rem:notseparating}
Let us assume that there exists a volume form $\vol$ on $M$, i.e.~that $M$ is oriented.
We make two comments which will be of importance later in this work.
First, consider $s,s^\prime\in \sect{M}{\bfA}$ such that the following equation holds, for 
all $\varphi\in\sectn{M}{\bfA^\dagger}$,
\begin{flalign}
\int_M \vol~\varphi(s) = \int_M\vol~\varphi(s^\prime)~.
\end{flalign}
It follows that $s=s^\prime$, i.e.~$\sectn{M}{\bfA^\dagger}$ is separating on $\sect{M}{\bfA}$, even when integrated.
Second, let $\varphi\in\sectn{M}{\bfA^\dagger}$ be such that the following equation holds true,
for all $s\in \sect{M}{\bfA}$,
\begin{flalign}\label{eqn:argumenttemp}
\int_M\vol~\varphi(s) =0~.
\end{flalign}
We fix any section $\widehat{s}\in \sect{M}{\bfA}$ and set $s-\widehat{s}=:\sigma$.
Then (\ref{eqn:argumenttemp}) equivalently gives the condition, for all $\sigma\in\sect{M}{\bfV}$,
\begin{flalign}
\int_M\vol~\varphi(\widehat{s}) + \int_M\vol~\varphi_{\bfV^\ast}(\sigma) =0~.
\end{flalign}
This implies that the integral $\int_M\vol~\varphi(\widehat{s}) =0$ has to vanish and  that 
the linear part $\varphi_{\bfV^\ast}=0$ is zero. As a consequence,
$\varphi = a\,\1$, where $a\in C^\infty_0(M)$ is a compactly supported function, which lies in the
kernel of the integral, i.e.~$\int_M\vol\,a =0$.
By $\1$ we denote the section of $\sect{M}{\bfA^\dagger}$ which associates to every point $x\in M$
the distinguished element (the constant affine map $\1_x:\bfA\vert_x\to\bbR$) 
in the fibre $\bfA^\dagger\vert_x$,
i.e.~$\1(x) :=\1_x$.
Thus, $\sect{M}{\bfA}$ is {\it not} separating on $\sectn{M}{\bfA^\dagger}$ when integrated, and
the degeneracy is characterized by the vector space of compactly supported functions with vanishing integral,
$\Ker(\int_M\vol~):=\big\{a\in C^\infty_0(M) : \int_M\vol~a=0\big\}$.
\end{rem}
\sk
  
To close this subsection we investigate the functoriality of assigning the
$C_0^\infty(M)$-module of compactly supported sections $\sectn{M}{\bfA^\dagger}$ to an affine bundle
$\big(M,(\bfA,M,\pi,A),(\bfV,M,\rho,V)\big)$.
\begin{defi}
The category $\Mod$ consists of the following objects and morphisms:
\begin{itemize}
\item An object in $\Mod$ is a pair $(\AA,\MM)$, where $\AA$ is an algebra and
$\MM$ is a right $\AA$-module.
\item A morphism between two objects $(\AA_1,\MM_1)$ and $(\AA_2,\MM_2)$ in $\Mod$
is a pair $(\xi,\Xi)$ consisting of an algebra homomorphism $\xi:\AA_1\to\AA_2$ and a linear 
map $\Xi:\MM_1\to \MM_2$, such that for all $s\in\MM_1$ and $a\in\AA_1$,
$\Xi(s\cdot a) = \Xi(s)\cdot\xi(a)$.
\end{itemize}
\end{defi}
\sk

We now define the functor $\DualConf:\AffBund\to \Mod$, which we interpret as associating
to an affine bundle the dual of the configuration space. 
\begin{lem}
There is a covariant functor $\DualConf:\AffBund\to\Mod$. It is specified on objects
by $\DualConf(M,\bfA,\bfV) = \big(C_0^\infty(M),\sectn{M}{\bfA^\dagger}\big)$ and on morphisms
by $\DualConf(f,\underline{f}) = \big(\underline{f}_\ast,{f^\dagger}_\ast\big)$, where 
$\underline{f}_\ast:C^\infty_0(M_1) \to C^\infty_0(M_2)$ is the
push-forward  of compactly supported functions, for all $a\in C^\infty_0(M_1)$ and $x\in M_2$,
\begin{subequations}
\begin{flalign}\label{eqn:pushforwardfunct}
 \big(\underline{f}_\ast(a)\big)(x):=\begin{cases}
a\big(\underline{f}^{-1}(x)\big) & ,~\text{if }x\in \underline{f}[M_1]~,\\
0 & ,~\text{else}~,
\end{cases} 
\end{flalign}
and ${f^\dagger}_\ast:\sectn{M_1}{\bfA^\dagger_1}\to \sectn{M_2}{\bfA^\dagger_2}$ 
is the push-forward of compactly supported sections, 
for all $\varphi\in\sectn{M_1}{\bfA^\dagger_1}$ and $x\in M_2$,
\begin{flalign}\label{eqn:pushforward}
 \big({f^\dagger}_\ast(\varphi)\big)(x):=\begin{cases}
f^\dagger\big(\varphi(\underline{f}^{-1}(x))\big) & ,~\text{if }x\in \underline{f}[M_1]~,\\
0 & ,~\text{else}~.
\end{cases}
\end{flalign}
\end{subequations}
\end{lem}
\begin{proof}
The pair $\big(\underline{f}_\ast,{f^\dagger}_\ast\big)$ is a morphism in $\Mod$, for all $\varphi\in\sectn{M_1}{\bfA^\dagger_1} $,
$a\in C^\infty_0(M_1)$ and $x\in M_2$,
\begin{flalign}
\big({f^\dagger}_\ast(\varphi\,a)\big)(x)=\begin{cases}
f^\dagger\big(\varphi(\underline{f}^{-1}(x))\,a(\underline{f}^{-1}(x))\big) & ,~\text{if }x\in \underline{f}[M_1]~,\\
0 & ,~\text{else}~.
\end{cases} = \big({f^\dagger}_\ast(\varphi)\,\underline{f}_\ast(a)\big)(x)~,
\end{flalign}
since $f^\dagger$ is a fibre-wise linear map. The functor properties are easily verified by using
Lemma \ref{lem:dualbundfunctor}.
\end{proof}

%%%%%%%%%%%%%%%%%%%%%%%%%%%%%%%%%%%%%%%%%%%%%%%%

\subsection{\label{sec:notation4}Differential operators on vector bundles}
Let $(\bfV_i,M,\rho_i,V_i)$, $i=1,2$, be two  vector bundles over the same $n$-dimensional base space $M$. 
A {\bf differential operator} is a linear map $P:\sect{M}{\bfV_1} \to \sect{M}{\bfV_2}$,
which in local coordinates $(x_1,\dots, x_n)$ and a local vector bundle chart of $\bfV_1$ and $\bfV_2$ looks
like
\begin{flalign}
P= \sum\limits_{\vert\alpha\vert \leq k} P^\alpha(x)\,\frac{\partial^{\vert\alpha\vert}}{\partial x^\alpha}~.
\end{flalign}
Here $k\in\bbN_0$ is the order of $P$, $\alpha \in \bbN_0^{n}$ is a multi-index and 
$\vert\alpha\vert  = \alpha_1+\dots +\alpha_n$ is its length.
Furthermore, $P^\alpha$ are local functions  with values in the linear homomorphisms
from $V_1$ to $V_2$ (the typical fibres)
and $\frac{\partial^{\vert\alpha\vert}}{\partial x^\alpha} := 
\frac{\partial^{\vert\alpha\vert}}{\partial x_1^{\alpha_1}\cdots \partial x_n^{\alpha_n}}$.

Let $M$ be equipped with a volume form $\vol$. Any differential operator is formally 
adjoinable to a differential operator (of the same order)
$P^\ast: \sect{M}{\bfV_2^\ast} \to\sect{M}{\bfV_1^\ast}$ (see e.g.~\cite[Proposition 1.2.12 and Theorem 1.2.15]{Waldmann} 
for a proof). 
Here $\bfV_i^\ast$ denotes the dual vector bundle
 of $\bfV_i$, $i=1,2$. The formal adjoint operator is uniquely specified by the equation,
 for all $\sigma \in \sect{M}{\bfV_1}$ and $\varphi\in\sect{M}{\bfV_2^\ast}$ with compact overlapping support,
 \begin{flalign}
 \int_M\vol ~\varphi\big(P(\sigma)\big) = \int_M\vol~\big(P^\ast(\varphi)\big)(\sigma)~.
 \end{flalign}
 
 Let us now consider a vector bundle $(\bfV,M,\rho,V)$ which is endowed with a nondegenerate bilinear form
 $\ip{~}{~}$. We can canonically identify $\bfV^\ast$ with $\bfV$ and we say that 
 a differential operator $P:\sect{M}{\bfV}\to \sect{M}{\bfV}$ is {\bf formally self-adjoint} (with respect to $\ip{~}{~}$)
 if its formal adjoint agrees with $P$ under this identification. Explicitly, we say that $P$ is formally self-adjoint
 if, for all $\sigma,\sigma^\prime\in\sect{M}{\bfV}$ with compact overlapping support,
 \begin{flalign}\label{eqn:formadjbilinform}
 \int_M\vol~\ip{\sigma}{P(\sigma^\prime)} = \int_M\vol~\ip{P(\sigma)}{\sigma^\prime}~.
 \end{flalign}

%%%%%%%%%%%%%%%%%%%%%%%%%%%%%%%%%%%%%%%%%%%%%%%%

\subsection{\label{sec:notation5}Globally hyperbolic spacetimes and Green-hyperbolic operators\\ on vector bundles}
We refer to \cite{Bar:2007zz,Bar:2011iu,Waldmann} for a more detailed exposition of the content of this subsection.

A {\bf Lorentzian manifold} is an oriented manifold $M$ equipped with a Lorentzian metric $g$
of signature $(-,+,\dots,+)$. 
We sometimes denote for ease of notation the Lorentzian manifold $(M,g)$ simply by $M$.
We denote the associated volume form by $\vol$.
A time-oriented Lorentzian manifold will be called a {\bf spacetime}.
For any subset $S\subseteq M$ we denote the causal future/past of
$S$ in $M$ by $J^\pm_M(S)$. 
 A subset $S\subseteq M$ is called {\bf causally compatible}, if
 $J^\pm_S(\{x\}) = J^\pm_M(\{x\})\cap S$, for all $x\in S$.
 We say that a closed subset $S\subseteq M$ is {\bf spacelike compact}, if there exists
 a compact subset $K\subseteq M$ such that $S\subseteq J_M(K):= J_M^+(K)\cup J_M^-(K)$.
 
 A {\bf Cauchy surface} is a subset $\Sigma\subset M$, which is met exactly once by any 
 inextensible timelike curve. We say that a spacetime is {\bf globally hyperbolic},
 if it contains a Cauchy surface. We shall require the following theorem
 proven by Bernal and S{\'a}nchez \cite{Bernal:2004gm,Bernal:2005qf}.
 \begin{theo}\label{theo:bernalsanchez}
 Let $(M,g)$ be a globally hyperbolic spacetime.  
 \begin{itemize}
 \item[(i)] Then there exists a manifold $\Sigma$, a smooth one-parameter family
 of Riemannian metrics $g_t$, $t\in\bbR$, on $\Sigma$ and a smooth positive function $\vartheta$ on $\bbR\times\Sigma$,
 such that $(M,g)$ is isometric to $(\bbR\times \Sigma,-\vartheta \,dt^2 \oplus g_t)$. Each
 $\{t\}\times \Sigma$ corresponds to a smooth Cauchy surface in $(M,g)$.
 \item[(ii)] Let also $\widetilde{\Sigma}\subset M$ be a smooth spacelike Cauchy surface in $(M,g)$.
 Then there exists a smooth splitting $(M,g)\simeq (\bbR\times \Sigma,-\vartheta\,dt^2\oplus g_t)$
 as in (i) such that $\widetilde{\Sigma}$ corresponds to $\{0\}\times \Sigma$.
 \end{itemize}
 \end{theo}
\sk
\begin{defi}
Let $P:\sect{M}{\bfV}\to \sect{M}{\bfV}$ be a differential operator on a vector bundle $\bfV$ over a spacetime
$M$. A {\bf retarded/advanced Green's operator} for $P$ is a linear map $G^\pm: \sectn{M}{\bfV} \to\sect{M}{\bfV}$
satisfying
\begin{itemize}
\item[(i)] $P\circ G^\pm =\id_{\sectn{M}{\bfV}}$,
\item[(ii)] $G^\pm\circ P\big\vert_{\sectn{M}{\bfV}} = \id_{\sectn{M}{\bfV}}$,
\item[(iii)] $\supp\big(G^\pm(h)\big) \subseteq J^\pm_M\big(\supp(h)\big)$, for any $h\in\sectn{M}{\bfV}$.
\end{itemize}
\end{defi}
\sk
\begin{defi}\label{eqn:greenhypop}
Let $P:\sect{M}{\bfV}\to \sect{M}{\bfV}$ be a differential operator on a vector bundle $\bfV$ over a globally hyperbolic
spacetime $M$. We say that $P$ is a {\bf Green-hyperbolic operator}, if there exist retarded/advanced Green's operators
for $P$ and its formal adjoint $P^\ast$.
\end{defi}
\sk

By Remark 3.7.~in \cite{Bar:2011iu} the Green's operators are necessarily unique\footnote{
The proof of this statement as presented in \cite{Bar:2011iu}
requires also the existence of Green's operators for $P^\ast$. This is why we have modified
our Definition \ref{eqn:greenhypop} accordingly. We are grateful to Ko Sanders for pointing us out this issue.
}. Examples of Green-hyperbolic operators are wave operators (also called normally hyperbolic operators)
and operators of Dirac-type.

We finally review two important statements on properties of Green's operators. See Lemma 3.3 and Theorem 3.5 in \cite{Bar:2011iu}
for the proofs.
\begin{lem}\label{lem:greenadjoint}
Let $M$ be a globally hyperbolic spacetime and $\bfV$ a vector bundle over $M$ endowed with 
a nondegenerate bilinear form $\ip{~}{~}$. Let further $G^\pm$ be the
 retarded/advanced Green's operators for a Green-hyperbolic operator $P:\sect{M}{\bfV}\to \sect{M}{\bfV}$.
 Then the retarded/advanced Green's operators ${G^\ast}^\pm$ for the formal adjoint $P^\ast:\sect{M}{\bfV}\to \sect{M}{\bfV}$
 defined with respect to $\ip{~}{~}$ satisfy, 
 for all $h,k\in\sectn{M}{\bfV}$,
  \begin{flalign}
  \int_M\vol~\ip{{G^\ast}^\pm ( h )}{k}= \int_M\vol~\ip{h}{G^\mp(k)}~.
  \end{flalign}
\end{lem}
\sk
\begin{theo}\label{theo:complex}
Let $M$ be a globally hyperbolic spacetime and $\bfV$ a vector bundle over $M$.
Let further $P:\sect{M}{\bfV}\to \sect{M}{\bfV}$ be a Green-hyperbolic operator with retarded/advanced Green's operators
$G^\pm:\sectn{M}{\bfV}\to\sect{M}{\bfV}$. Setting
$G:=G^+-G^- : \sectn{M}{\bfV}\to\sectsc{M}{\bfV}$, where $\sectsc{M}{\bfV}$ denotes the space of
 sections of spacelike compact support,
the following sequence of linear maps is a complex, which is exact everywhere:
\begin{flalign}
\{0\} \stackrel{~}{\longrightarrow} \sectn{M}{\bfV} \stackrel{P}{\longrightarrow} \sectn{M}{\bfV} 
\stackrel{G}{\longrightarrow}\sectsc{M}{\bfV}\stackrel{P}{\longrightarrow} \sectsc{M}{\bfV}~
\end{flalign}
\end{theo}
\sk

%%%%%%%%%%%%%%%%%%%%%%%%%%%%%%%%%%%%%%%%%%%%%%%%
%%%%%%%%%%%%%%%%%%%%%%%%%%%%%%%%%%%%%%%%%%%%%%%%

\section{\label{sec:affinediffop}Affine differential operators}
Let $\big(M,(\bfA_1,M,\pi_1,A_1),(\bfV_1,M,\rho_1,V_1)\big)$ be an affine bundle
and $(\bfV_2,M,\rho_2,V_2)$ a vector bundle over
the same oriented  manifold $M$. We denote the volume form on $M$ by $\vol$.
\begin{defi}
\begin{itemize}
\item[(i)] An {\bf affine differential operator} is an affine map
$P:\sect{M}{\bfA_1}\to \sect{M}{\bfV_2}$, such that its linear part
$P_\bfV:\sect{M}{\bfV_1}\to \sect{M}{\bfV_2}$ is a  differential operator.
\item[(ii)] We call an affine differential operator $P:\sect{M}{\bfA_1}\to \sect{M}{\bfV_2}$ {\bf formally adjoinable}
if there exists a differential operator $P^\ast: \sect{M}{\bfV^\ast_2} \to \sect{M}{\bfA_1^\dagger}$, such that
for all $\varphi\in \sectn{M}{\bfV^\ast_2}$  and $s\in\sect{M}{\bfA_1}$,
\begin{flalign}\label{eqn:adjoin}
\int_M\vol~\varphi\big(P(s)\big)= \int_M\vol~\big(P^\ast(\varphi)\big)(s)~.
\end{flalign}
\end{itemize}
\end{defi}
\sk

\begin{rem}
Notice that we demand the target space of $P$ to be the vector space of sections of a vector bundle.
The reason is that all $P$ in our work are used for specifying an equation of motion $P(s) =0$ on $s\in\sect{M}{\bfA_1}$,
which requires the zero section on the right hand side.
\end{rem}
\sk
\begin{ex}\label{ex:inhomlin}
Let $(\bfV_i,M,\rho_i,V_i)$, $i=1,2$, be two vector bundles over the same manifold $M$ and let
$P_\bfV:\sect{M}{\bfV_1}\to \sect{M}{\bfV_2}$ be a differential operator. 
Due to Example \ref{ex:vectorbundasaffinebund} we can regard $(\bfV_1,M,\rho_1,V_1)$  as an affine bundle, denoted
by $\big(M,(\bfA_1,M,\pi_1,A_1),(\bfV_1,M,\rho_1,V_1)\big)$. Given any section $J\in \sect{M}{\bfV_2}$, we can define the map
\begin{flalign}
P: \sect{M}{\bfA_1} \to \sect{M}{\bfV_2} ~,~~s\mapsto P(s) = P_\bfV(s) + J~,
\end{flalign}
which is an affine differential operator with linear part $P_\bfV$. If we interpret $J$ as a source term
for the linear differential operator $P_\bfV$, we observe that inhomogeneous differential operators between vector bundles
can be regarded as affine differential operators. 
For a more detailed discussion of field theories in this class of models see Section \ref{sec:inhommat}.
\end{ex}
\begin{ex}
As explained in Example \ref{ex:sequenceaffinebund}, the splittings of the Atiyah sequence of a principal bundle 
are described in terms of an affine bundle (the bundle of connections). For principal $U(1)$-bundles 
(i.e.~electromagnetism) one can prove that Maxwell's equation is an affine differential operator
from the bundle of connections to its underlying vector bundle. Details can be found in our future work \cite{BDSMaxwell}.
\end{ex}
\begin{theo}\label{theo:adjoin}
Every affine differential operator $P:\sect{M}{\bfA_1}\to \sect{M}{\bfV_2}$ is formally adjoinable.
However, the formal adjoint differential operator is not unique. Given two formal adjoints 
 $P^\ast,\widetilde{P}^\ast:
\sect{M}{\bfV^\ast_2} \to \sect{M}{\bfA^\dagger_1}$ of $P$, their difference  is given by $\widetilde{P}^\ast - P^\ast = \1\,Q$.
Here $Q:\sect{M}{\bfV^\ast_2} \to C^\infty(M)$ is a differential operator, such that,
for all $\varphi\in \sectn{M}{\bfV^\ast_2}$, $\int_M\vol~Q(\varphi)=0$, and $\1\in\sect{M}{\bfA_1^\dagger}$
is the canonical section mapping any point $x\in M$ to the distinguished affine map $\1_x\in \bfA_1^\dagger\vert_x$ in the fibre.
\end{theo}
\begin{proof}
Fix an arbitrary $\widehat{s}\in\sect{M}{\bfA_1}$. We consider the left hand side of (\ref{eqn:adjoin})
and obtain for the integrand $\varphi\big(P(s)\big) = \varphi\big(P(\widehat{s})\big) 
+ \varphi\big(P_\bfV(s-\widehat{s})\big)$.
The differential operator $P_\bfV$ is formally adjoinable to a differential operator $P^\ast_\bfV:\sect{M}{\bfV_2^\ast}\to
\sect{M}{\bfV_1^\ast}$ and the left hand side of (\ref{eqn:adjoin}) becomes
\begin{flalign}
 \int_M\vol~\varphi\big(P(s)\big)&= \int_M\vol~\Big(\varphi\big(P(\widehat{s})\big)  + \big(P_\bfV^\ast(\varphi)\big)(s-\widehat{s})\Big)~.
\end{flalign}
We define, for all $\varphi\in\sect{M}{\bfV^\ast_2}$, the map $P^\ast(\varphi):\sect{M}{\bfA_1}\to C^\infty(M)$ by, 
for all $s\in\sect{M}{\bfA_1}$,
\begin{flalign}\label{eqn:temp}
\big(P^\ast(\varphi)\big)(s):= \varphi\big(P(\widehat{s})\big)+ \big(P^\ast_\bfV(\varphi)\big)(s-\widehat{s})~.
\end{flalign}
Notice that, for all $\varphi\in\sect{M}{\bfV^\ast_2}$, the map $P^\ast(\varphi)$ is an affine map with linear part 
$P_\bfV^\ast(\varphi)$.

To show that $P^\ast$ is a differential operator, we fix an arbitrary point $x\in M$ and consider an open
neighbourhood  $U\subset M$ of $x$, such that
on $U$ there is an affine bundle chart of $\bfA_1$ and a vector bundle chart of $\bfV_2$. 
For the local sections we then have the isomorphisms $\sect{U}{\bfA_1} \simeq C^\infty(U,A_1)$ and
$\sect{U}{\bfV_2} \simeq C^\infty(U,V_2)$.
Analogously, using the induced vector bundle charts on $\bfA^\dagger_1$ and $\bfV^\ast_2$
we have the isomorphisms $\sect{U}{\bfA_1^\dagger}\simeq C^\infty(U,A_1^\dagger)$ and $\sect{U}{\bfV_2^\ast}
\simeq C^\infty(U,V_2^\ast)$.
A $C^\infty(U)$-module basis of $C^\infty(U,A^\dagger_1)$
 is given by, for all $x\in U$,  $e_0(x):= \1_x$ and $e_b(x)$ defined by, for all $a\in A_1$,
$\big(e_b(x)\big)(a) := e_b^\ast\big(a-\widehat{s}(x)\big)$, where $e_b^\ast$ denotes a basis of $V_1^\ast$
and the index $b$ runs from $1$ to the dimension of $V_1^\ast$.
Expressing (\ref{eqn:temp}) in terms of this basis and using the basis expansion $P^\ast_\bfV(\varphi)\vert_{U}^{} = 
\big(P^\ast_\bfV(\varphi)\big)^b\,e_b^\ast $ (sum over $b$ understood)
we obtain
\begin{flalign}\label{eqn:temp2}
P^\ast(\varphi)\vert_U^{} = \varphi\big(P(\widehat{s})\big) \,e_0 ~ + ~ \big(P^\ast_\bfV(\varphi)\big)^b\,e_b~.
\end{flalign}
Using also a $C^\infty(U)$-module basis $\{\tilde e^\ast_\beta\}$ of $C^\infty(U,V_2^\ast)$,
 we can write $\varphi\vert_U^{}= \varphi^\beta\,\tilde e^\ast_\beta $ (sum over $\beta$ understood)
 and (\ref{eqn:temp2}) reads
\begin{flalign}
P^\ast(\varphi)\vert_U^{} = \tilde e_\beta^\ast \big(P(\widehat{s})\big) \,\varphi^\beta\,e_0 ~+~(P^\ast_\bfV)^b_\beta\, \varphi^\beta~e_b~.
\end{flalign}
This shows that $P^\ast$ is a  differential operator, since the coefficients in the matrix $(P^\ast_\bfV)^b_\beta$
and in the vector $\tilde e_\beta^\ast \big(P(\widehat{s})\big) $ are differential operators acting on $\varphi^\beta$.

According to Remark \ref{rem:notseparating}, there are potential issues in uniquely defining
the operator $P^\ast$ via the integral equation (\ref{eqn:adjoin}). Let us assume that there
is a second differential operator $\widetilde{P}^\ast: \sect{M}{\bfV^\ast_2} \to \sect{M}{\bfA_1^\dagger}$
satisfying (\ref{eqn:adjoin}). As a consequence of Remark \ref{rem:notseparating}, there
exists for each $\varphi\in \sectn{M}{\bfV^\ast_2}$ an element $a_\varphi\in\Ker\big(\int_M\vol~\big)$,
such that $\widetilde{P}^\ast(\varphi) - P^\ast(\varphi) = a_\varphi\,\1$.
Since the left hand side is a differential operator acting on $\varphi$, the right hand side also has to be a differential
operator acting on $\varphi$, which we denote by $Q:\sect{M}{\bfV^\ast_2} \to C^\infty(M)$. We obtain
 $\widetilde{P}^\ast(\varphi) - P^\ast(\varphi) = Q(\varphi)\,\1$. The integral condition
on $a_\varphi$ implies that $\int_M\vol~Q(\varphi)=0$, for all $\varphi\in \sectn{M}{\bfV^\ast_2}$.
\end{proof}
\begin{rem}
Examples of differential operators $Q$ characterizing the non-uniqueness of the formal adjoint 
are given by differential operators of divergence form.
However, notice that the linear part of $P^\ast(\varphi)\in \sect{M}{\bfA_1^\dagger}$ is unique, 
for all $\varphi\in\sect{M}{\bfV^\ast_2}$.
\end{rem}

%%%%%%%%%%%%%%%%%%%%%%%%%%%%%%%%%%%%%%%%%%%%%%%%
%%%%%%%%%%%%%%%%%%%%%%%%%%%%%%%%%%%%%%%%%%%%%%%%

\section{\label{sec:classaff}Classical affine matter field theory}
We consider an affine bundle $\big(M,(\bfA,M,\pi,A),(\bfV,M,\rho,V)\big)$ of which the underlying vector bundle is equipped 
with a nondegenerate bilinear form $\ip{~}{~}$.
Using $\ip{~}{~}$ we can identify $\bfV^\ast$ with $\bfV$.
We assume $M$ to be a globally hyperbolic spacetime.
\begin{defi}
An {\bf affine Green-hyperbolic  operator} is an affine differential operator
$P:\sect{M}{\bfA}\to \sect{M}{\bfV}$, such that its linear part 
$P_\bfV:\sect{M}{\bfV}\to \sect{M}{\bfV}$ is a Green-hyperbolic operator.
\end{defi}
\sk

A formal adjoint of an affine Green-hyperbolic operator is 
a differential operator $P^\ast:\sect{M}{\bfV^\ast} \to \sect{M}{\bfA^\dagger}$.
Identifying $\bfV^\ast$ with $\bfV$ by using $\ip{~}{~}$, we can also regard
it as a differential operator (denoted with a slight abuse of notation by the same symbol) 
$P^\ast:\sect{M}{\bfV} \to \sect{M}{\bfA^\dagger}$ satisfying, for all $h\in\sectn{M}{\bfV}$
and $s\in\sect{M}{\bfA}$,
\begin{flalign}\label{eqn:adjointgreenhypop}
\int_M\vol~\ip{h}{P(s)}= \int_M\vol~\big(P^\ast(h)\big)(s)~.
\end{flalign}
As in Theorem \ref{theo:adjoin} the formal adjoint of $P$ is not uniquely specified by (\ref{eqn:adjointgreenhypop}).
Given two differential operators $P^\ast,\widetilde{P}^\ast:\sect{M}{\bfV} \to \sect{M}{\bfA^\dagger}$ 
satisfying (\ref{eqn:adjointgreenhypop}),
then their difference is given by $\widetilde{P}^\ast-P^\ast = \1 \,Q$, where $Q:\sect{M}{\bfV}\to C^\infty(M)$
is a differential operator, such that $\int_M\vol~Q(h)=0$, for all $h\in \sectn{M}{\bfV}$.
We denote by 
\begin{flalign}
\mathrm{FormAdj}_P := \big\{P^\ast: \sect{M}{\bfV} \to \sect{M}{\bfA^\dagger} : P^\ast \text{ is formal adjoint of } P\big\}
\end{flalign}
the set of all formal adjoints of $P$.
\sk

A simple class of local observables on $\sect{M}{\bfA}$ is characterized by the $C^\infty(M)$-module
$\sectn{M}{\bfA^\dagger}$. Indeed, given any $\varphi\in \sectn{M}{\bfA^\dagger}$ we define the corresponding
observable by the local affine  functional
\begin{flalign}\label{eqn:laf}
F_\varphi:\sect{M}{\bfA} \to \bbR\,,~s\mapsto F_\varphi(s) = \int_M \vol~\varphi(s)~.
\end{flalign}
According to Remark \ref{rem:notseparating} there are trivial observables: 
Let $\varphi\in \sectn{M}{\bfA^\dagger}$ be of the form $\varphi = a\,\1$, with $a\in\Ker\big(\int_M\vol~\big)$.
Then $F_\varphi(s)=0$, for all $s\in\sect{M}{\bfA}$. We denote by
\begin{flalign}
\mathrm{Triv} := \Big\{a\,\1\in \sectn{M}{\bfA^\dagger} : a\in C^\infty_0(M) \text{ satisfies } \int_M\vol~a=0\Big\}\subset \sectn{M}{\bfA^\dagger}
\end{flalign}
the vector space of trivial observables. Hence, we can label the nontrivial observables by elements $\varphi$ of the vector space
$\sectn{M}{\bfA^\dagger}/\mathrm{Triv}$, where we do not use brackets to denote equivalence classes for notational 
simplicity. The linear part $\varphi_\bfV$ does not depend on the choice of representative in the equivalence class
 $\varphi\in \sectn{M}{\bfA^\dagger}/\mathrm{Triv}$.

The vector space of observables $\{F_\varphi: \varphi\in\sectn{M}{\bfA^\dagger}/\mathrm{Triv}\}$ is separating,
i.e.~from $F_\varphi(s)=F_{\varphi}(s^\prime)$, for all $\varphi\in\sectn{M}{\bfA^\dagger}/\mathrm{Triv}$, it follows that
$s=s^\prime$.
Notice that an element $s\in\sect{M}{\bfA}$ is a solution of $P(s)=0$, where $P$ is
an affine Green-hyperbolic operator, if and only if $F_\varphi(s)=0$ for all
$\varphi$ given by  $\varphi=P^\ast(h)$, with $h\in\sectn{M}{\bfV}$. These 
$\varphi\in \sectn{M}{\bfA^\dagger}/\mathrm{Triv}$ do not depend on the
 choice $P^\ast\in\mathrm{FormAdj}_P$ of a formal adjoint of $P$:
Given $P^\ast,\widetilde{P}^\ast\in\mathrm{FormAdj}_P$, then the difference $\widetilde{P}^\ast(h)-P^\ast(h)
=Q(h)\,\1\in\mathrm{Triv}$ is equivalent to zero, since we are working in $\sectn{M}{\bfA^\dagger}/\mathrm{Triv}$.
We denote this vector subspace by $P^\ast\big [ \sectn{M}{\bfV} \big] \subset \sectn{M}{\bfA^\dagger}/\mathrm{Triv}$.
\sk

We can now associate the classical phase space to this theory. The functoriality of this construction
is investigated in the next section.
\begin{propo}\label{propo:sympl}
Let $P:\sect{M}{\bfA}\to \sect{M}{\bfV}$ be an affine Green-hyperbolic operator such that the linear part 
$P_\bfV:\sect{M}{\bfV}\to \sect{M}{\bfV}$ is formally self-adjoint with respect to
$\ip{~}{~}$. Denote the retarded/advanced Green's operator for $P_\bfV$ by
$G^\pm_{\bfV}: \sectn{M}{\bfV}\to \sect{M}{\bfV}$ and set $G_\bfV := G^+_{\bfV}-G^-_{\bfV}$.

Then the vector space $\EE := \big(\sectn{M}{\bfA^\dagger}/\mathrm{Triv}\big)/P^\ast\big[\sectn{M}{\bfV} \big] $ can be equipped
with the bilinear map
\begin{flalign}\label{eqn:tau}
\tau : \EE\times \EE\to \bbR~,~~([\varphi],[\psi])\mapsto \tau([\varphi],[\psi])=\int_M\vol~\ip{\varphi_\bfV}{G_\bfV (\psi_\bfV)}~,
\end{flalign}
where $\varphi_\bfV,\psi_\bfV\in\sectn{M}{\bfV}$ denote the linear parts of $\varphi,\psi\in \sectn{M}{\bfA^\dagger}/\mathrm{Triv}$
specified by $\varphi(s+\sigma) = \varphi(s) +\ip{ \varphi_\bfV}{\sigma}$,
 for all $s\in\sect{M}{\bfA}$ and $\sigma\in \sect{M}{\bfV}$
(analogously for $\psi$).
\end{propo}
\begin{proof}
Remember that the linear part of $\varphi\in \sectn{M}{\bfA^\dagger}/\mathrm{Triv}$ does not depend on the choice of
representative.
Thus, it remains to show that $\tau(\varphi,\psi) = \tau(\psi,\varphi)=0$,
for all $\varphi\in \sectn{M}{\bfA^\dagger}/\mathrm{Triv}$,
 if $\psi=P^\ast(h)$ with $h \in \sectn{M}{\bfV} $. From the proof of Theorem \ref{theo:adjoin} we know that
the linear part of $\psi$ is given by $\psi_\bfV = P^\ast_\bfV(h) = P_\bfV(h)$, where in the last step we have
also used that $P_\bfV$ is formally self-adjoint. The proof follows from the fact that
$G_\bfV$ is formally skew-adjoint and $G_\bfV\circ P_\bfV =0$ on $\sectn{M}{\bfV}$, see 
Lemma \ref{lem:greenadjoint} and Theorem \ref{theo:complex}.
\end{proof}
\begin{rem}
The form of $\tau$ given in (\ref{eqn:tau}) can also be motivated by considering Peierls brackets on the 
vector space of local affine functionals $\{ F_\varphi: \varphi\in \sectn{M}{\bfA^\dagger}/\mathrm{Triv}\}$, see (\ref{eqn:laf}).
We obtain for the derivative  of $F_\varphi$ along $\sigma\in\sect{M}{\bfV}$, for all $s\in\sect{M}{\bfA}$,
\begin{flalign}
\delta_\sigma F_\varphi(s) := \lim\limits_{\epsilon \searrow 0} \frac{1}{\epsilon} \big(F_\varphi(s+\epsilon\sigma)-F_\varphi(s)\big)
= \int_M \vol~\ip{\varphi_\bfV}{\sigma}~,
\end{flalign}
with $\varphi_\bfV\in\sectn{M}{\bfV}$ being the linear part of $\varphi$. Notice that in the derivatives only
the linear part enters. Defining the Peierls bracket via the evaluation
of the causal propagator $G_\bfV$ of $P_\bfV$ we  obtain
\begin{flalign}
\big\{F_\varphi,F_\psi\big\} = \int_M\vol~\ip{\varphi_\bfV}{G_\bfV(\psi_\bfV)}~\oone~~,
\end{flalign}
which agrees with (\ref{eqn:tau}). In this expression $\oone:\sect{M}{\bfA}\to \bbR$ denotes the constant functional. 
This provides an explanation why it is natural that in the map $\tau$ of
Proposition \ref{propo:sympl} only the linear part enters.
\end{rem}
\sk

The bilinear map $\tau$ of Proposition \ref{propo:sympl} is degenerate. We denote by
$\mathcal{L}\subseteq \EE$ (the left null space of $\tau$) the vector subspace of all elements $[\varphi]$ satisfying
$\tau([\varphi],[\psi]) =0$, for all $[\psi]\in\EE$. Similarly, we denote by $\mathcal{R}\subseteq \EE$
 (the right null space of $\tau$) the vector subspace of all elements 
 $[\varphi]$ satisfying $\tau([\psi],[\varphi])=0$, for all $[\psi]\in \EE$.
We can give a characterization of the vector subspaces $\mathcal{L}$ and $\mathcal{R}$.
\begin{propo}\label{propo:degen}
An element $[\varphi]\in \EE$ is in $\mathcal{R}$, if and only if there exists a representative 
 $\varphi\in\sectn{M}{\bfA^\dagger}$ of $[\varphi]$ which is of the form $\varphi = a\,\1$, with $a\in C^\infty_0(M)$.
It further holds true that $\mathcal{L} =\mathcal{R}=:\mathcal{N}$ (the null space of $\tau$). 
 \end{propo}
\begin{proof}
If $\varphi = a\,\1$ with some $a\in C^\infty_0(M)$, then the linear part $\varphi_\bfV =0$ and thus
$\tau\big([\psi],[\varphi]\big) =0$, for all $[\psi]\in\EE$.

To show the other direction, let $[\varphi]\in\mathcal{R}$, i.e.~$\tau\big([\psi],[\varphi]\big) =0$, for all $[\psi]\in\EE$. 
It follows from (\ref{eqn:tau}) that
\begin{flalign}\label{eqn:degtemp}
\int_M\vol~\ip{\psi_\bfV}{G_\bfV(\varphi_\bfV)} =0~,
\end{flalign}
for all linear parts $\psi_\bfV\in \sectn{M}{\bfV}$. For each element $h\in\sectn{M}{\bfV}$
we can define, via fixing an arbitrary section $\widehat{s}\in\sect{M}{\bfA}$, an element 
$\psi\in\sectn{M}{\bfA^\dagger}$ by setting $\psi(s):= \ip{h}{s-\widehat{s}}$, for all $s\in\sect{M}{\bfA}$.
The linear part of this $\psi$ is then $h$, and thus (\ref{eqn:degtemp}) implies that
$G_\bfV(\varphi_\bfV)=0$. As a consequence of $P_\bfV$ being a Green-hyperbolic operator (cf.~Theorem \ref{theo:complex}),
 $\varphi_\bfV \in P_\bfV\big[\sectn{M}{\bfV}\big]$, i.e.~there exists an $h\in\sectn{M}{\bfV}$ 
 such that $\varphi_\bfV = P_\bfV(h)$.
Remember that the linear part of $P^\ast(h)$ is (independently on the choice of $P^\ast\in\mathrm{FormAdj}_P$)
given by $P^\ast(h)_\bfV = P_\bfV(h)$. Thus, subtracting $\varphi - P^\ast(h) =:\varphi^\prime$
we get a representative $\varphi^\prime\in\sectn{M}{\bfA^\dagger}$ of the same equivalence class $[\varphi]$ which has a 
trivial linear part, $\varphi^\prime_\bfV =0$. Since the linear part is trivial, $\varphi^\prime$ has to be of the
form $\varphi^\prime = a\,\1$, with some $a\in C^\infty_0(M)$.

Since $G_\bfV$ is formally skew-adjoint by Lemma \ref{lem:greenadjoint} we have $\mathcal{L}=\mathcal{R}=:\mathcal{N}$.
\end{proof}
\begin{cor}\label{cor:nullspace}
The linear map
\begin{flalign}
\iota: C^\infty_0(M)/\Ker\left(\int_M\vol~\right) \to \mathcal{N}~,~~a\mapsto [a\,\1]
\end{flalign}
is a vector space isomorphism.
\end{cor}
\begin{proof}
The map $\iota$ is well-defined and linear. By Proposition \ref{propo:degen} it is surjective.
We show that it is injective: Assume that $[a\,\1] = [a^\prime\,\1]\in\EE$. Then there is an $h\in\sectn{M}{\bfV}$, such that
$a\,\1 - a^\prime\,\1 = P^\ast(h)$ in $\sectn{M}{\bfA^\dagger}/\mathrm{Triv}$. (Since we work on
$\sectn{M}{\bfA^\dagger}/\mathrm{Triv}$, the choice of $P^\ast\in \mathrm{FormAdj}_P$ does not matter.)
Taking the linear part of this equation we find $P_\bfV(h)=0$, which due to the Green-hyperbolicity of $P_\bfV$ implies 
$h=0$ (cf.~Theorem \ref{theo:complex}).
Thus, $a\,\1 = a^\prime\,\1$ in $\sectn{M}{\bfA^\dagger}/\mathrm{Triv}$, which implies 
$a-a^\prime \in\Ker\big(\int_M\vol~\big)$.
\end{proof}
\sk

%%%%%%%%%%%%%%%%%%%%%%%%%%%%%%%%%%%%%%%%%%%%%%%%

\section{\label{sec:category}Categorical formulation}
We formulate in the language of category theory the association of the phase space $(\EE,\tau)$ in 
Proposition \ref{propo:sympl}.  Let us define the following category:
\begin{defi}\label{def:geometriccategory}
The category $\GlobHypAffGreen$ consists of the following objects and morphisms:
\begin{itemize}
\item An object in $\GlobHypAffGreen$ is a triple $(M,\bfA,P)$, where
\begin{itemize}
\item $M$ is a globally hyperbolic spacetime,
\item $\bfA$ is an affine bundle over $M$, whose underlying vector bundle $\bfV$ is endowed with
a nondegenerate bilinear form $\ip{~}{~}$,
\item $P:\sect{M}{\bfA} \to \sect{M}{\bfV}$ is an affine Green-hyperbolic operator, whose linear part $P_\bfV$
is formally self-adjoint with respect to $\ip{~}{~}$.
\end{itemize}
\item A morphism between two objects $(M_1,\bfA_1,P_1)$ and $(M_2,\bfA_2,P_2)$ in $\GlobHypAffGreen$ 
is a morphism $(f,\underline{f})$ in the category $\AffBund$, such that 
$(f_\bfV,\underline{f})$ (see Lemma \ref{lem:linearpartfunctor}) preserves the bilinear forms and
\begin{itemize}
\item $\underline{f}:M_1\to M_2$ is an orientation and time-orientation preserving isometric embedding with
$\underline{f}[M_1]\subseteq M_2$ causally compatible and open,
\item the following diagram commutes
\begin{flalign}\label{eqn:Pcompatible}
\xymatrix{
\sect{M_2}{\bfA_2} \ar[rr]^-{P_2} \ar[d]_-{f^\ast}& & \sect{M_2}{\bfV_2} \ar[d]^-{{f_\bfV}^\ast}\\
\sect{M_1}{\bfA_1} \ar[rr]^-{P_1} & & \sect{M_1}{\bfV_1}
}
\end{flalign}
where the pull-backs are defined by $f^\ast(s) := f^{-1} \circ s \circ \underline{f}$, for all
$s\in \sect{M_2}{\bfA_2}$, and ${f_\bfV}^\ast(\sigma) := f_\bfV^{-1}\circ \sigma\circ \underline{f}$,
for all $\sigma\in\sect{M_2}{\bfV_2}$.
\end{itemize}
\end{itemize}
\end{defi}
\sk

Given a morphism $(f,\underline{f})$ between two objects $(M_1,\bfA_1,P_1)$ and $(M_2,\bfA_2,P_2)$ in 
the category $\GlobHypAffGreen$, we also define the push-forward
${f_\bfV}_\ast:\sectn{M_1}{\bfV_1} \to\sectn{M_2}{\bfV_2}$ of compactly supported sections,
\begin{flalign}
{f_\bfV}_\ast(\sigma)(x):=\begin{cases}
f_\bfV\big(\sigma(\underline{f}^{-1}(x))\big) & ,~\text{if } x\in \underline{f}[M_1]~,\\
0 & ,~\text{else}~.
\end{cases}
\end{flalign}
We obtain the following
\begin{lem}\label{lem:Greenres}
Let $(f,\underline{f})$ be a morphism between two objects 
$(M_1,\bfA_1,P_1)$ and $(M_2,\bfA_2,P_2)$ in $\GlobHypAffGreen$.
Let further $G^\pm_{1_\bfV}$ and $G^\pm_{2_\bfV}$ be the Green's operators
for the linear parts $P_{1_\bfV}$ and $P_{2_\bfV}$, respectively. Then
$G_{1_\bfV}^\pm = {f_\bfV}^\ast \circ G_{2_\bfV}^\pm \circ {f_\bfV}_\ast$.
\end{lem}
\begin{proof}
The first part of the proof coincides with the one of Lemma 3.2 in \cite{Bar:2011iu}. For completeness, we review the steps.
Let us define $\widetilde{G}^\pm_{1_\bfV}:= {f_\bfV}^\ast\circ G_{2_\bfV}^\pm \circ {f_\bfV}_\ast$
and fix an arbitrary $h\in\sectn{M_1}{\bfV_1}$. We obtain
\begin{flalign}
\nn\supp\big(\widetilde{G}^\pm_{1_\bfV}(h)\big)&= \underline{f}^{-1}\big[\supp\big(G^\pm_{2_\bfV}({f_\bfV}_\ast(h))\big)\big]\\
\nn &\subseteq \underline{f}^{-1}\big[ J^\pm_{M_2}\big(\supp({f_\bfV}_\ast(h))\big)\big]\\
\nn & = \underline{f}^{-1}\big[ J^\pm_{M_2}\big(\underline{f}[\supp(h)]\big)\big]\\
 &= J^\pm_{M_1}\big(\supp(h)\big)~,
\end{flalign}
where in the last line we have used that, by definition, $\underline{f}[M_1] \subseteq M_2$ is causally compatible.

For proving that $\widetilde{G}^\pm_{1_\bfV} = G^\pm_{1_\bfV}$, it remains to show that
$P_{1_\bfV}\circ \widetilde{G}^\pm_{1_\bfV} = \id$ and $\widetilde{G}^\pm_{1_\bfV}\circ P_{1_\bfV} =\id$ on
$\sectn{M_1}{\bfV_1}$. The desired result is then a consequence of the uniqueness of the Green's operators for $P_{1_\bfV}$.
The commutative diagram (\ref{eqn:Pcompatible}) induces the commutative diagram
\begin{flalign}\label{eqn:PVcompatible}
\xymatrix{
\sect{M_2}{\bfV_2} \ar[rr]^-{P_{2_\bfV}} \ar[d]_-{{f_\bfV}^\ast}& & \sect{M_2}{\bfV_2} \ar[d]^-{{f_\bfV}^\ast}\\
\sect{M_1}{\bfV_1} \ar[rr]^-{P_{1_\bfV}} & & \sect{M_1}{\bfV_1}
}~
\end{flalign}
since  the linear part of ${f}^\ast$ is ${f_\bfV}^\ast$.
Hence, we have on $\sectn{M_1}{\bfV_1}$
\begin{subequations}
\begin{flalign}
P_{1_\bfV}\circ \widetilde{G}^\pm_{1_\bfV}  = P_{1_\bfV}\circ {f_\bfV}^\ast \circ G^\pm_{2_\bfV} \circ {f_\bfV}_\ast = {f_\bfV}^\ast\circ P_{2_\bfV}\circ  G^\pm_{2_\bfV} \circ {f_\bfV}_\ast =\id~
\end{flalign}
and 
\begin{flalign}
\nn \widetilde{G}^\pm_{1_\bfV} \circ P_{1_\bfV} &= {f_\bfV}^\ast \circ G^\pm_{2_\bfV} \circ {f_\bfV}_\ast  \circ P_{1_\bfV} = {f_\bfV}^\ast \circ G^\pm_{2_\bfV} \circ {f_\bfV}_\ast  \circ P_{1_\bfV}\circ {f_\bfV}^\ast\circ {f_\bfV}_\ast\\
&
={f_\bfV}^\ast \circ G^\pm_{2_\bfV} \circ {f_\bfV}_\ast  \circ{f_\bfV}^\ast\circ P_{2_\bfV}\circ  {f_\bfV}_\ast={f_\bfV}^\ast \circ G^\pm_{2_\bfV} \circ P_{2_\bfV}\circ {f_\bfV}_\ast = \id~.
\end{flalign}
\end{subequations}
Here we have also used the identities ${f_\bfV}^\ast \circ {f_\bfV}_\ast =\id$ on
all of $\sectn{M_1}{\bfV_1}$ and
${f_\bfV}_\ast\circ {f_\bfV}^\ast= \id$ on the subset $\sectn{\underline{f}[M_1]}{\bfV_2}\subseteq \sectn{M_2}{\bfV_2}$.
\end{proof}
\begin{defi}
The category $\VecBiLin$ consists of the following objects and morphisms:
\begin{itemize}
\item An object in $\VecBiLin$ is a tuple $(\EE,\tau)$, where $\EE$ is a
(possibly infinite-dimensional) vector space and $\tau:\EE\times \EE \to \bbR$ is a bilinear map.
\item A morphism between two objects $(\EE_1,\tau_1)$ and $(\EE_2,\tau_2)$ in $\VecBiLin$
is an {\it injective} linear map $L:\EE_1\to\EE_2$, which preserves the bilinear maps, 
i.e.~$\tau_2(L(v),L(w)) = \tau_1(v,w)$, for all $v,w\in \EE_1$.
\end{itemize}
\end{defi}
\sk

The association of the phase space $(\EE,\tau)$ of Proposition \ref{propo:sympl} is functorial.
We call this functor the phase space functor $\PhaseSpace$.
\begin{theo}\label{theo:PhaseSpace}
There is a covariant functor $\PhaseSpace: \GlobHypAffGreen \to \VecBiLin$. 
It associates to any object $(M,\bfA,P)$ in $\GlobHypAffGreen$ the object 
$\PhaseSpace(M,\bfA,P):= (\EE,\tau)$ in $\VecBiLin$,
 which has been constructed in Proposition \ref{propo:sympl}.
Given a morphism $(f,\underline{f})$ between two objects $(M_1,\bfA_1,P_1)$ and $(M_2,\bfA_2,P_2)$ in $\GlobHypAffGreen$,
the functor associates a morphism  in $\VecBiLin$ as follows
\begin{flalign}\label{eqn:sympmorphism}
\PhaseSpace(f,\underline{f}):\PhaseSpace(M_1,\bfA_1,P_1) \to \PhaseSpace(M_2,\bfA_2,P_2) ~,~~[\varphi]\mapsto [{f^\dagger}_\ast(\varphi)]~,
\end{flalign}
where ${f^\dagger}_\ast: \sectn{M_1}{\bfA^\dagger_1} \to \sectn{M_2}{\bfA^\dagger_2}$
 is the push-forward given in (\ref{eqn:pushforward}).
\end{theo}
\begin{proof}
We have to show that (\ref{eqn:sympmorphism}) is well-defined, which consists of two steps:
First, we have to prove that ${f^\dagger}_\ast$ maps $\mathrm{Triv}_1$ into $\mathrm{Triv}_2$
and thus gives rise to a map (denoted with a slight abuse of notation by the same symbol)
${f^\dagger}_\ast: \sectn{M_1}{\bfA^\dagger_1} /\mathrm{Triv}_1 \to \sectn{M_2}{\bfA^\dagger_2}/\mathrm{Triv}_2$.
Second, we have to show that  this 
${f^\dagger}_\ast$ maps $P_1^\ast[\sectn{M_1}{\bfV_1}]$ to $P_2^\ast[\sectn{M_2}{\bfV_2}]$.

For the first point, let $\varphi \in\mathrm{Triv}_1$ be arbitrary, i.e.~$\varphi = a\,\1_1$ with some
 $a\in\Ker\big(\int_{M_1}\volo~\big)$.
The push-forward gives ${f^\dagger}_\ast(\varphi) = \underline{f}_\ast(a)\,\1_2$,
where $\underline{f}_\ast(a)$ denotes the push-forward  of compactly supported functions 
(see (\ref{eqn:pushforwardfunct})). We obtain that
 ${f^\dagger}_\ast(\varphi)\in\mathrm{Triv}_2$, since
 \begin{flalign}
 \int_{M_2}\volt~\underline{f}_\ast(a) = \int_{\underline{f}[M_1]}\volt~ \underline{f}_\ast(a) = \int_{M_1}\volo~a =0~.
 \end{flalign}

To show the second point, let $h\in\sectn{M_1}{\bfV_1}$ and $s\in\sect{M_2}{\bfA_2}$ be arbitrary. We have
\begin{flalign}
\nn \int_{M_2}\volt~\Big({f^\dagger}_\ast\big(P_1^\ast(h)\big)\Big)(s)& ~\,= \int_{M_1} \volo~\big(P_1^\ast(h)\big)(f^\ast(s))
=\int_{M_1} \volo~\ip{h}{P_1\big(f^\ast(s)\big)}_1^{~}\\
\nn & \stackrel{\text{(\ref{eqn:Pcompatible})}}{=}\int_{M_1} \volo~\ip{h}{{f_\bfV}^\ast\big(P_2(s)\big)}_1^{~}
= \int_{M_2} \volt~\ip{{f_\bfV}_\ast(h)}{P_2(s)}_2^{~} \\
&~\,= \int_{M_2}\volt~\Big(P_2^\ast\big({f_\bfV}_\ast(h)\big)\Big)(s)~.
\end{flalign}
Thus, we have the compatibility condition ${f^\dagger}_\ast \circ P_1^\ast = P_2^\ast\circ {f_\bfV}_\ast$,
if we regard both maps from  $\sectn{M_1}{\bfV_1}$ to $\sectn{M_2}{\bfA^\dagger_2}/\mathrm{Triv}_2$. 
This implies that the map  (\ref{eqn:sympmorphism}) is well-defined.

We check if the map (\ref{eqn:sympmorphism})  preserves the bilinear maps $\tau$.
Let $[\varphi],[\psi]\in \PhaseSpace(M_1,\bfA_1,P_1)$ be arbitrary. We obtain from (\ref{eqn:tau})
\begin{flalign}
 \tau_2([{f^\dagger}_\ast(\varphi)],[{f^\dagger}_\ast(\psi)]) = 
\int_{M_2}\volt~\ip{{f^\dagger}_\ast(\varphi)_\bfV}{G_{2_\bfV}({f^\dagger}_\ast(\psi)_\bfV)}_2^{}~.
\end{flalign}
The linear part of the push-forward ${f^\dagger}_\ast(\varphi)$ is given by the push-forward
 ${f_\bfV}_\ast(\varphi_\bfV)$ of the linear part of $\varphi$ (the same applies to $\psi$).
Thus,
\begin{flalign}
\nn \tau_2([{f^\dagger}_\ast(\varphi)],[{f^\dagger}_\ast(\psi)])  &=  \int_{M_2}\volt~\ip{{f_\bfV}_\ast(\varphi_\bfV)}{G_{2_\bfV}({f_\bfV}_\ast(\psi_\bfV))}_2^{}\\
\nn &= \int_{M_1}\volo~\ip{\varphi_\bfV}{{f_\bfV}^\ast\big(G_{2_\bfV}({f_\bfV}_\ast(\psi_\bfV))\big)}_1^{}\\
&= \int_{M_1}\volo~\ip{\varphi_\bfV}{G_{1_\bfV}(\psi_\bfV)}_1^{}= \tau_1([\varphi],[\psi])~,
\end{flalign}
where we have used Lemma \ref{lem:Greenres} in the last line.

It remains to be proven the injectivity of the map (\ref{eqn:sympmorphism}), i.e.~we have to show that
${f^\dagger}_\ast(\varphi) = P_2^\ast(h)$, for some $h\in\sectn{M_2}{\bfV_2}$, only if
$\varphi = P_1^\ast(k)$, for some $k\in \sectn{M_1}{\bfV_1}$. 
To study the properties of $h$ let us consider the linear part 
${f_\bfV}_\ast(\varphi_\bfV) = P_{2_\bfV}(h)$ of the equation ${f^\dagger}_\ast(\varphi) = P_2^\ast(h)$.
Since ${f_\bfV}_\ast(\varphi_\bfV)$ and $h$ are by definition of compact support in $M_2$,
and $P_{2_\bfV}$ is Green-hyperbolic, we can use the Green's operators $G^\pm_{2_\bfV}$ for $P_{2_\bfV}$
to obtain $h=G^+_{2_\bfV}\big({f_\bfV}_\ast(\varphi_\bfV)\big)=G^-_{2_\bfV}\big({f_\bfV}_\ast(\varphi_\bfV)\big)$.
Applying the pull-back ${f_\bfV}^\ast$ to $h$ we find ${f_\bfV}^\ast(h) =
{f_\bfV}^\ast\big(G^\pm_{2_\bfV}\big({f_\bfV}_\ast(\varphi_\bfV)\big)\big) = G^\pm_{1_\bfV}(\varphi_\bfV)$,
 where in the last equality we have used Lemma \ref{lem:Greenres}. Since $\varphi_\bfV$ is of compact support in $M_1$
 and $M_1$ is globally hyperbolic, also the support of ${f_\bfV}^\ast(h)$ is compact in $M_1$.
 We thus can apply ${f_\bfV}_\ast$ and write
 $h= h-{f_\bfV}_\ast \big({f_\bfV}^\ast(h)\big)+{f_\bfV}_\ast \big({f_\bfV}^\ast(h)\big) =: 
 \chi + {f_\bfV}_\ast \big({f_\bfV}^\ast(h)\big)$. Notice the support property $\chi(x) =0$,
  for all $x\in \underline{f}[M_1]$.
 We obtain for our initial $\varphi$ the following expression
 \begin{flalign}
 \varphi = {f^\dagger}^\ast\big(P_2^\ast(h)\big) = {f^\dagger}^\ast\Big(P_2^\ast(\chi) + P_2^\ast\big({f_\bfV}_\ast \big({f_\bfV}^\ast(h)\big)\Big) = {f^\dagger}^\ast\big(P_2^{\ast}(\chi)\big)
 + P_1^\ast\big({f_\bfV}^\ast(h)\big)~.
 \end{flalign}
 Because of the support property of $\chi$ and the fact that $P_2^\ast$ is a differential operator, hence preserving the support,
 we obtain that ${f^\dagger}^\ast\big(P_2^\ast(\chi)\big) \equiv 0$ and
  thus $\varphi = P_1^\ast\big({f_\bfV}^\ast(h)\big)$.
 Injectivity follows by setting $k={f_\bfV}^\ast(h)$ .

The functor properties follow from the functoriality of assigning to a morphism
 $(f,\underline{f})$ in $\GlobHypAffGreen$ the
push-forward ${f^\dagger}_\ast$ defined in (\ref{eqn:pushforward}).
\end{proof}
\sk

We conclude this subsection by proving two important theorems on properties of the functor
$\PhaseSpace$.
\begin{theo}\label{theo:causalclass}
The covariant functor $\PhaseSpace: \GlobHypAffGreen \to \VecBiLin$ satisfies the classical causality property:\vspace{2mm}

Let $(M_i,\bfA_i,P_i)$, $i=1,2,3$, be objects in $\GlobHypAffGreen$
and $(f_j,\underline{f_j})$ morphisms from $(M_j,\bfA_j,P_j)$ to
$(M_3,\bfA_3,P_3)$, $j=1,2$, such that $\underline{f_1}[M_1]$ and $\underline{f_2}[M_2]$ are causally disjoint 
 in $M_3$. Then  $\tau_3$ acts trivially among the vector subspaces
 $\PhaseSpace(f_1,\underline{f_1})\big[\PhaseSpace(M_1,\bfA_1,P_1)\big]$
and $\PhaseSpace(f_2,\underline{f_2})\big[\PhaseSpace(M_2,\bfA_2,P_2)\big]$ of $\PhaseSpace(M_3,\bfA_3,P_3)$,
i.e.~for all $[\varphi]\in \PhaseSpace(M_1,\bfA_1,P_1)$ and $[\psi]\in 
\PhaseSpace(M_2,\bfA_2,P_2)$,
\begin{subequations}
\begin{flalign}
\tau_3\big(\PhaseSpace(f_1,\underline{f_1})([\varphi]),\PhaseSpace(f_2,\underline{f_2})([\psi])\big) =0~,
\end{flalign}
and
\begin{flalign} \label{eqn:classcaustemp}
\tau_3\big(\PhaseSpace(f_2,\underline{f_2})([\psi]),\PhaseSpace(f_1,\underline{f_1})([\varphi])\big)=0~.
\end{flalign}
\end{subequations}
\end{theo}
\begin{proof}
From (\ref{eqn:sympmorphism}) and (\ref{eqn:tau}) we find
\begin{flalign}
\tau_3\big(\PhaseSpace(f_1,\underline{f_1})([\varphi]),\PhaseSpace(f_2,\underline{f_2})([\psi])\big) 
=\int_{M_3} \volthree~ \ip{{f_1^\dagger}_\ast(\varphi)_\bfV}{G_{3_\bfV}({f_2^\dagger}_\ast(\psi)_\bfV)}_3=0~,
\end{flalign}
since $\supp\big({f_1^\dagger}_\ast(\varphi)\big)\subseteq \underline{f_1}[M_1]$
and $\supp\big({f_2^\dagger}_\ast(\psi)\big)\subseteq \underline{f_2}[M_2]$
are causally disjoint, and thus 
$\supp\big({f_1^\dagger}_\ast(\varphi)\big) \cap \supp \big(G_{3_\bfV}({f_2^\dagger}_\ast(\psi))\big)=\emptyset$.
The proof of (\ref{eqn:classcaustemp}) follows by the same argument.
\end{proof}
\begin{theo}\label{theo:timesliceclass}
The covariant functor $\PhaseSpace: \GlobHypAffGreen \to \VecBiLin$ satisfies the classical time-slice axiom:\vspace{2mm}

Let $(M_j,\bfA_j,P_j)$, $j=1,2$, be objects in $\GlobHypAffGreen$ and $(f,\underline{f})$ a morphism
from $(M_1,\bfA_1,P_1)$ to $(M_2,\bfA_2,P_2)$  such that $\underline{f}[M_1]\subseteq M_2$  contains a Cauchy
surface of $M_2$. Then 
\begin{flalign}\label{eqn:timesliceiso}
\PhaseSpace(f,\underline{f}): \PhaseSpace(M_1,\bfA_1,P_1) \to \PhaseSpace(M_2,\bfA_2,P_2)
\end{flalign}
is an isomorphism.
\end{theo}
\begin{proof}
Injectivity of $\PhaseSpace(f,\underline{f})$ is given for any morphism $(f,\underline{f})$ by Theorem \ref{theo:PhaseSpace}.
It thus remains to prove surjectivity.

Let us take an arbitrary $[\varphi] \in \PhaseSpace(M_2,\bfA_2,P_2) $ and consider a
representative $\varphi\in\sectn{M_2}{\bfA^\dagger_2}$. If $\supp(\varphi)\subset \underline{f}[M_1]$,
we obtain by using the pull-back and push-forward maps
$\varphi = {f^\dagger}_\ast\big({f^\dagger}^\ast(\varphi)\big)$, with 
${f^\dagger}^\ast(\varphi)\in\sectn{M_1}{\bfA_1^\dagger}$. This gives 
$[\varphi] = \PhaseSpace(f,\underline{f})\big([{f^\dagger}^\ast(\varphi)]\big)$.
Thus, surjectivity would follow, if we could prove that for all $[\varphi] \in \PhaseSpace(M_2,\bfA_2,P_2) $ 
there is a representative $\varphi^\prime\in\sectn{M_2}{\bfA^\dagger_2}$, such that
$\supp(\varphi^\prime)\subset \underline{f}[M_1]$.

Given an arbitrary representative $\varphi\in\sectn{M_2}{\bfA^\dagger_2}$ we are thus looking
for  $h\in \sectn{M_2}{\bfV_2}$ and  $a\in\Ker\big(\int_{M_2}\volt~\big)$, such that
\begin{flalign}\label{eqn:temptimeslice}
\varphi^\prime = \varphi + P_2^\ast(h) + a\,\1_2
\end{flalign}
 has the required support property. By construction,
$[\varphi^\prime] = [\varphi]$ in $\PhaseSpace(M_2,\bfA_2,P_2)$. The choice of $P^\ast\in\mathrm{FormAdj}_P$
does not matter, since different choices just lead to a redefinition of $a$.
Notice that (\ref{eqn:temptimeslice}) implies the following equation for the linear parts
\begin{flalign}\label{eqn:temptimeslice2}
\varphi^\prime_\bfV = \varphi_\bfV + P_{2_\bfV}(h)~.
\end{flalign}
This is the same equation as it appears in linear matter field theories and thus an $h\in\sectn{M_2}{\bfV_2}$ can be found
such that $\varphi^\prime_\bfV$ has the required support property. For completeness we recall the proof.
 Let $\Sigma_2\subset \underline{f}[M_1] \subseteq M_2$
 be the Cauchy surface assumed in the hypotheses of this theorem. We denote by $\Sigma_1\subset M_1$
 the subset given by $\Sigma_1 := \underline{f}^{-1}[\Sigma_2]$. Then $\Sigma_1$ is a Cauchy surface
 in $M_1$ and by Theorem \ref{theo:bernalsanchez} there exist two more Cauchy surfaces $\Sigma_1^\pm\neq \Sigma_1$
 of $M_1$ in the future/past of $\Sigma_1$. Since $\underline{f}[M_1]\subseteq M_2$ is causally compatible,
 the subsets $\Sigma_2^\pm := \underline{f}[\Sigma_1^\pm]$ are Cauchy surfaces in $M_2$ and
 $J^-_{M_2}(\Sigma_2^+) \cap J^+_{M_2}(\Sigma_2^-) \subseteq \underline{f}[M_1] \subseteq M_2$.
The statement now follows if we can show that for any $\varphi_\bfV\in \sectn{M_2}{\bfV_2}$ 
there is an $h\in \sectn{M_2}{\bfV_2}$
such that $\varphi_\bfV^\prime$ defined in (\ref{eqn:temptimeslice2}) has support in 
$J^-_{M_2}(\Sigma_2^+) \cap J^+_{M_2}(\Sigma_2^-) $. We can assume without loss of generality that
$\varphi_\bfV$ has support in $J_{M_2}^+(\Sigma_2^-)$. Choosing any $\chi\in C^\infty(M_2)$, such that
$\chi\equiv 0$ on $J^-_{M_2}(\Sigma_2^-)$ and $\chi \equiv 1$ on $J^+_{M_2}(\Sigma_2^+)$
we define $h:= - \chi\, G_{2_\bfV}^-(\varphi_\bfV)$. Then $\varphi_\bfV^\prime$ defined in (\ref{eqn:temptimeslice2})
has support in $J^-_{M_2}(\Sigma_2^+) \cap J^+_{M_2}(\Sigma_2^-) $, which proves the statement.

 Let us now continue  investigating (\ref{eqn:temptimeslice}). By fixing 
 any section $\widehat{s}\in\sect{M_2}{\bfA_2}$ we can split $\varphi^\prime$
 into a constant part and the induced linear part, $\varphi^\prime = \varphi^\prime(\widehat{s})\,\1_2 + 
  \varphi^\prime_\bfV(\,\cdot\, -\widehat{s})$. Then (\ref{eqn:temptimeslice}) is equivalent to (\ref{eqn:temptimeslice2}) and 
 the following equation on $C^\infty_0(M_2)$
 \begin{flalign}
 \varphi^\prime(\widehat{s}) = \varphi(\widehat{s}) +\big( P_2^\ast(h)\big)(\widehat{s}) + a~.
 \end{flalign}
 Notice that we have already specified $h$ via  (\ref{eqn:temptimeslice2}), thus we now  use 
 $a\in \Ker\big(\int_{M_2}\volt~\big)$ in order to obtain a $\varphi^\prime(\widehat{s})$ of the
 required support. We make the ansatz
 \begin{flalign}
 a= -\varphi(\widehat{s}) - \big(P_2^\ast(h)\big)(\widehat{s}) + a^\prime~,
 \end{flalign}
 where $a^\prime \in C^\infty_0(\underline{f}[M_1])$ is a (positive or negative definite) 
 bump function localized in $\underline{f}[M_1]\subseteq M_2$.
 In order to have $a\in\Ker\big(\int_{M_2}\volt~\big)$, we can always adjust the height of the bump function, since all
 individual integrals exist. Hence, we also find $\varphi^\prime(\widehat{s})=a^\prime$ with the required support property.
 In total, $\varphi^\prime = a^\prime \,\1_2 + \varphi_\bfV^\prime(\,\cdot\, -\widehat{s})$ has compact support in
 $\underline{f}[M_1]\subseteq M_2$ and from the argument in the paragraph above we obtain that (\ref{eqn:timesliceiso})
 is bijective.
\end{proof}

%%%%%%%%%%%%%%%%%%%%%%%%%%%%%%%%%%%%%%%%%%%%%%%%

\section{\label{sec:quant}Bosonic and fermionic quantization}
For discussing the quantization of affine matter field theories we have to 
specialize to certain subcategories of $\GlobHypAffGreen$, which describe bosonic and, respectively, fermionic
field theories.

\begin{defi}
We define the subcategory $\GlobHypAffGreen^\mathrm{bos}$  of $\GlobHypAffGreen$ as follows:
\begin{itemize}
\item An object $(M,\bfA,P)$ in $\GlobHypAffGreen$ is an object in $\GlobHypAffGreen^\mathrm{bos}$,
if the nondegenerate bilinear form $\ip{~}{~}$ on the vector bundle $\bfV$ underlying $\bfA$ is {\bf symmetric}.
\item Given two objects $(M_1,\bfA_1,P_1)$ and $(M_2,\bfA_2,P_2)$ in $\GlobHypAffGreen^\mathrm{bos}$
we take as morphisms in $\GlobHypAffGreen^\mathrm{bos}$ all morphisms available in $\GlobHypAffGreen$.
\end{itemize}
We define the subcategory $\GlobHypAffGreen^\mathrm{ferm}$ of $\GlobHypAffGreen$ analogously to above
 by replacing symmetric with {\bf antisymmetric}.
\end{defi}
\sk

We denote by $\VecBiLin^{\mathrm{asym}}$ the following category:
An object in $\VecBiLin^{\mathrm{asym}}$  is given by a tuple $(\EE,\tau)$, where $\EE$ is 
a (possibly infinite-dimensional) vector space
and $\tau$ is an antisymmetric bilinear map $\tau:\EE\times\EE\to \bbR$.
A morphism between two objects $(\EE_1,\tau_1)$ and $(\EE_2,\tau_2)$ in $\VecBiLin^{\mathrm{asym}}$ is an injective linear map
$L:\EE_1\to\EE_2$ which preserves the bilinear maps $\tau$, i.e.~$\tau_2(Lv,Lw) = \tau_1(v,w)$, for all
$v,w\in\EE_1$. We further denote by $\VecBiLin^{\mathrm{sym}}$ the following category:
An object in $\VecBiLin^{\mathrm{sym}}$ is given by a tuple $(\EE,\tau)$, where $\EE$ is a (possibly infinite-dimensional)
vector space and $\tau$ is a symmetric bilinear map $\tau:\EE\times \EE\to \bbR$.
A morphism between two objects $(\EE_1,\tau_1)$ and $(\EE_2,\tau_2)$ in $\VecBiLin^{\mathrm{sym}}$ is an injective linear map
$L:\EE_1\to\EE_2$ which preserves the bilinear maps $\tau$.
Notice that we can regard $\VecBiLin^{\mathrm{asym}}$ and $\VecBiLin^{\mathrm{sym}}$ as subcategories
of  $\VecBiLin$.

\begin{lem}
The covariant functor $\PhaseSpace: \GlobHypAffGreen \to \VecBiLin$ constructed in Theorem \ref{theo:PhaseSpace} 
reduces to covariant functors $\PhaseSpace:\GlobHypAffGreen^\mathrm{bos} 
\to \VecBiLin^{\mathrm{asym}} $
and $\PhaseSpace:\GlobHypAffGreen^\mathrm{ferm} \to \VecBiLin^{\mathrm{sym}}$.
\end{lem}
\begin{proof}
We have to show that given any object $(M,\bfA,P)$ in $\GlobHypAffGreen^\mathrm{bos}$
the map $\tau$ given in (\ref{eqn:tau}) is antisymmetric. This is a consequence of the symmetry of $\ip{~}{~}$
and the fact that $G_\bfV$ is formally skew-adjoint. Analogously, one finds that 
for any object $(M,\bfA,P)$ in $\GlobHypAffGreen^\mathrm{ferm}$ the map $\tau$ in (\ref{eqn:tau})
is symmetric, since $\ip{~}{~}$ is antisymmetric.
\end{proof}

Using the covariant functors $\CCR: \VecBiLin^{\mathrm{asym}} \to \astAlg$  and $\CAR: \VecBiLin^{\mathrm{sym}} \to \astAlg$ 
discussed in the Appendix \ref{app:CCRandCAR}, we obtain the covariant functors
\begin{subequations}
\begin{flalign}
\QFT^\mathrm{bos} &:= \CCR\circ \PhaseSpace:  \GlobHypAffGreen^\mathrm{bos} \to \astAlg~,\\
\QFT^\mathrm{ferm} &:= \CAR\circ \PhaseSpace:  \GlobHypAffGreen^\mathrm{ferm} \to \astAlg~.
\end{flalign}
\end{subequations}
\begin{theo}\label{theo:boslqft}
The covariant functor $\QFT^\mathrm{bos} :  \GlobHypAffGreen^\mathrm{bos} \to \astAlg$ is a bosonic
locally covariant quantum field theory, i.e.~the following properties hold true:
\begin{itemize}
\item[(i)] Let $(M_i,\bfA_i,P_i)$, $i=1,2,3$, be objects in $\GlobHypAffGreen^\mathrm{bos}$
and $(f_j,\underline{f_j})$ be morphisms from $(M_j,\bfA_j,P_j)$ to $(M_3,\bfA_3,P_3)$,  $j=1,2$,
such that $\underline{f_1}[M_1]$ and $\underline{f_2}[M_2]$ are causally disjoint in $M_3$.
Then $\QFT^{\mathrm{bos}}(f_1,\underline{f_1})\big[\QFT^\mathrm{bos}(M_1,\bfA_1,P_1)\big]$
and $\QFT^{\mathrm{bos}}(f_2,\underline{f_2})\big[\QFT^\mathrm{bos}(M_2,\bfA_2,P_2)\big]$
commute as subalgebras of $\QFT^\mathrm{bos}(M_3,\bfA_3,P_3)$.
\item[(ii)] Let $(M_j,\bfA_j,P_j)$, $j=1,2$, be objects in $\GlobHypAffGreen^\mathrm{bos}$ and $(f,\underline{f})$ a morphism
from $(M_1,\bfA_1,P_1)$ to $(M_2,\bfA_2,P_2)$  such that $\underline{f}[M_1]\subseteq M_2$  contains a Cauchy
surface of $M_2$. Then 
\begin{flalign}\label{eqn:quantumtimesliceisobos}
\QFT^\mathrm{bos}(f,\underline{f}): \QFT^\mathrm{bos}(M_1,\bfA_1,P_1) \to \QFT^\mathrm{bos}(M_2,\bfA_2,P_2)
\end{flalign}
is an isomorphism.
\end{itemize}
\end{theo}
\begin{proof}
Proof of (i):  Notice that the subalgebra
$\QFT^{\mathrm{bos}}(f_j,\underline{f_j})\big[\QFT^\mathrm{bos}(M_j,\bfA_j,P_j)\big]$
is isomorphic to the subalgebra of $\QFT^\mathrm{bos}(M_3,\bfA_3,P_3)$ generated by
 $\PhaseSpace(f_j,\underline{f_j})\big[\PhaseSpace(M_j,\bfA_j,P_j)\big]$, for $j=1,2$.
Then by Theorem \ref{theo:causalclass} all
generators of $\QFT^{\mathrm{bos}}(f_1,\underline{f_1})\big[\QFT^\mathrm{bos}(M_1,\bfA_1,P_1)\big]$
are commuting with all generators of $\QFT^{\mathrm{bos}}(f_2,\underline{f_2})\big[\QFT^\mathrm{bos}(M_2,\bfA_2,P_2)\big]$, 
which implies that general elements are commuting as well.

Proof of (ii): By Theorem \ref{theo:timesliceclass} we know that $\PhaseSpace(M_1,\bfA_1,P_1)$
and $\PhaseSpace(M_2,\bfA_2,P_2)$ are isomorphic by $\PhaseSpace(f,\underline{f})$. Since $\CCR$ is a covariant functor (in particular
it is compatible with the composition of morphisms),
the induced morphism in $\astAlg$ given by $\QFT^\mathrm{bos}(f,\underline{f}) = \CCR\big(\PhaseSpace(f,\underline{f})\big)$
is actually an isomorphism.
\end{proof}
\sk

\begin{theo}
The covariant functor $\QFT^\mathrm{ferm} :  \GlobHypAffGreen^\mathrm{ferm} \to \astAlg$ is a fermionic
locally covariant quantum field theory, i.e.~the following properties hold true:
\begin{itemize}
\item[(i)]  Let $(M_i,\bfA_i,P_i)$, $i=1,2,3$, be objects in $\GlobHypAffGreen^\mathrm{ferm}$
and $(f_j,\underline{f_j})$ be morphisms from $(M_j,\bfA_j,P_j)$ to $(M_3,\bfA_3,P_3)$,  $j=1,2$,
such that $\underline{f_1}[M_1]$ and $\underline{f_2}[M_2]$ are causally disjoint in $M_3$.
Then $\QFT^{\mathrm{ferm}}(f_1,\underline{f_1})\big[\QFT^\mathrm{ferm}(M_1,\bfA_1,P_1)\big]$
and $\QFT^{\mathrm{ferm}}(f_2,\underline{f_2})\big[\QFT^\mathrm{ferm}(M_2,\bfA_2,P_2)\big]$ 
 super-commute\footnote{
This means that the $\mathbb{Z}_2$-odd parts of the algebra anti-commute and the $\mathbb{Z}_2$-even parts
commute.
} as subalgebras of $\QFT^\mathrm{ferm}(M_3,\bfA_3,P_3)$.
\item[(ii)] Let $(M_j,\bfA_j,P_j)$, $j=1,2$, be objects in $\GlobHypAffGreen^\mathrm{ferm}$ and $(f,\underline{f})$ a morphism
from $(M_1,\bfA_1,P_1)$ to $(M_2,\bfA_2,P_2)$  such that $\underline{f}[M_1]\subseteq M_2$  contains a Cauchy
surface of $M_2$. Then 
\begin{flalign}\label{eqn:quantumtimesliceisoferm}
\QFT^\mathrm{ferm}(f,\underline{f}): \QFT^\mathrm{ferm}(M_1,\bfA_1,P_1) \to \QFT^\mathrm{ferm}(M_2,\bfA_2,P_2)
\end{flalign}
is an isomorphism.
\end{itemize}
\end{theo}
\begin{proof}
This theorem is proven by following similar steps as in the proof of 
Theorem \ref{theo:boslqft}.
\end{proof}

%%%%%%%%%%%%%%%%%%%%%%%%%%%%%%%%%%%%%%%%%%%%%%%%

\section{\label{sec:linfunctor}Linearization functor}
We study in this section a canonical association of linear matter field theories to affine ones.
A suitable category for describing the former is given in \cite[Definition 3.1]{Bar:2011iu}, which we
are going to review now (compare with Definition \ref{def:geometriccategory} of the present paper).
We will slightly adapt this definition to be consistent with
 the conventions used throughout our work.
\begin{defi}
The category $\mathsf{GlobHypLinGreen}$ consists of the following objects and morphisms:
\begin{itemize}
\item An object in $\mathsf{GlobHypLinGreen}$ is a triple $(M,\bfV,P)$, where
\begin{itemize}
\item $M$ is a globally hyperbolic spacetime,
\item $\bfV$ is a vector bundle over $M$ endowed with
a nondegenerate bilinear form $\ip{~}{~}$,
\item $P:\sect{M}{\bfV} \to \sect{M}{\bfV}$ is a formally self-adjoint Green-hyperbolic operator.
\end{itemize}
\item A morphism between two objects $(M_1,\bfV_1,P_1)$ and $(M_2,\bfV_2,P_2)$ in $\mathsf{GlobHypLinGreen}$ 
is a morphism $(f,\underline{f})$ in the category $\VecBund$ (preserving the bilinear forms), such that
\begin{itemize}
\item $\underline{f}:M_1\to M_2$ is an orientation and time-orientation preserving isometric embedding with
$\underline{f}[M_1]\subseteq M_2$ causally compatible and open,
\item the following diagram commutes
\begin{flalign}\label{eqn:Plincompatible}
\xymatrix{
\sect{M_2}{\bfV_2} \ar[rr]^-{P_2} \ar[d]_-{{f}^\ast}& & \sect{M_2}{\bfV_2} \ar[d]^-{{f}^\ast}\\
\sect{M_1}{\bfV_1} \ar[rr]^-{P_1} & & \sect{M_1}{\bfV_1}
}
\end{flalign}
where the pull-back is defined by ${f}^\ast(\sigma) := f^{-1}\circ \sigma\circ \underline{f}$,
for all $\sigma\in\sect{M_2}{\bfV_2}$.
\end{itemize}
\end{itemize}
\end{defi}
\sk
It follows from the work in \cite{Bar:2011iu} that there exists, similarly to Theorem \ref{theo:PhaseSpace},
a covariant phase space functor $\PhaseSpace_{\mathrm{lin}} $ from $\mathsf{GlobHypLinGreen}$ to $\VecBiLin$.
It associates to any object $(M,\bfV,P)$ in  $\mathsf{GlobHypLinGreen}$  the object
$\PhaseSpace_{\mathrm{lin}}(M,\bfV,P)= (\EE_\lin,\tau_\lin)$ in $\VecBiLin$, where 
$\EE_\lin := \Gamma^\infty_0(M,\bfV)/P\big[\Gamma^\infty_0(M,\bfV)\big]$
and
 \begin{flalign}\label{eqn:taulin}
\tau_\lin: \EE_\lin\times \EE_\lin\to \bbR~,~~([h],[k])\mapsto \tau_\lin([h],[k]) = \int_M \vol~\ip{h}{G(k)}~.
\end{flalign}
Notice that in the linear case the (left and right) null space of $\tau_\lin$ is trivial, 
i.e.~$0=\mathcal{L}_\lin = \mathcal{R}_\lin =: \mathcal{N}_\lin$. 
To any morphism  $(f,\underline{f})$ between two objects $(M_1,\bfV_1,P_1)$
 and $(M_2,\bfV_2,P_2)$ in $\mathsf{GlobHypLinGreen}$,
the functor associates the following morphism  in $\VecBiLin$
\begin{flalign}\label{eqn:sympmorphismlin}
\PhaseSpace_\lin(f,\underline{f}):\PhaseSpace_\lin(M_1,\bfV_1,P_1) \to \PhaseSpace_\lin(M_2,\bfV_2,P_2) ~,~~[h]\mapsto
 [f_\ast(h)]~,
\end{flalign}
where $f_\ast: \sectn{M_1}{\bfV_1} \to \sectn{M_2}{\bfV_2}$
 is the push-forward corresponding to the vector bundle map $(f,\underline{f})$.
\sk\sk

The linear part functor of Lemma \ref{lem:linearpartfunctor} provides a covariant functor from the category
$\GlobHypAffGreen$ to the category
$\mathsf{GlobHypLinGreen}$ (that with a slight abuse of notation we denote by the same symbol),
 which we interpret as a {\it geometric} linearization of the affine theory.
\begin{propo}\label{propo:linearmatterfunctor}
%For each object $(M,\bfA,P)$ in $\GlobHypAffGreen$, consider only the underlying vector bundle and
%the linear part of the affine Green-hyperbolic operator, namely set $\mathfrak{LinBund}(M,\bfA,P) = (M,\bfV,P_\bfV)$.
%For each morphism $(f,\underline{f})$ in $\GlobHypAffGreen$, consider only the underlying vector bundle map,
%namely set $\mathfrak{LinBund}(f,\underline{f}) = (f_\bfV,\underline{f})$.
%Then $\mathfrak{LinBund}$ is a covariant functor.
There is a covariant functor $\mathfrak{LinBund}: \GlobHypAffGreen \to \mathsf{GlobHypLinGreen}$.
It is specified on objects by $\mathfrak{LinBund}(M,\bfA,P) = (M,\bfV,P_\bfV)$
and on morphisms by $\mathfrak{LinBund}(f,\underline{f}) = (f_\bfV,\underline{f})$.
\end{propo}
\begin{proof}
For any object $(M,\bfA,P)$ in $\GlobHypAffGreen$, $\mathfrak{LinBund}(M,\bfA,P) = (M,\bfV,P_\bfV)$
is an object in $\mathsf{GlobHypLinGreen}$: Indeed, $M$ is a globally hyperbolic spacetime,
$\bfV$ a vector bundle over $M$ with a nondegenerate bilinear form $\ip{~}{~}$
and $P_\bfV$ is a formally self-adjoint Green-hyperbolic operator.

For any morphism $(f,\underline{f})$ in $\GlobHypAffGreen$ the tuple $\mathfrak{LinBund}(f,\underline{f}) =(f_\bfV,\underline{f})$
is a morphism in $\mathsf{GlobHypLinGreen}$: From Lemma \ref{lem:linearpartfunctor} we know that
$(f_\bfV,\underline{f})$ is morphism in $\VecBund$. By definition, $(f_\bfV,\underline{f})$ preserves the nondegenerate
bilinear forms and $\underline{f}$ is an orientation and time-orientation preserving isometric embedding with
$\underline{f}[M_1]\subseteq M_2$ causally compatible and open. The required commutative diagram
(\ref{eqn:Plincompatible}) was shown to hold true in Lemma \ref{lem:Greenres}, see in particular the
diagram (\ref{eqn:PVcompatible}).
\end{proof}
\sk

The next step would be to find a covariant functor $\mathfrak{LinPhSp}$ from $\VecBiLin$ to itself,
such that  we can close a ``commutative diagram of functors'' (what we actually find below is a natural isomorphism)
\begin{flalign}\label{eqn:wishlist}
\xymatrix{
\GlobHypAffGreen\ar[d]_-{\PhaseSpace} \ar[rr]^-{\mathfrak{LinBund}} & & \mathsf{GlobHypLinGreen}\ar[d]^-{\PhaseSpace_\lin}\\
\VecBiLin \ar[rr]^-{\mathfrak{LinPhSp}} & & \VecBiLin
}\qquad\quad \text{ (heuristic!)}
\end{flalign}
This would show that linearization at the {\it geometric} level (via the functor $\mathfrak{LinBund}$) 
and a subsequent association of the phase space to the resulting linear theory is equivalent to linearization
of the phase space of affine theories at the {\it algebraic} level (via the functor $\mathfrak{LinPhSp}$).
In other words, diagram (\ref{eqn:wishlist}) would show that the affine matter field theory can be consistently linearized.
\sk

For carrying out this program we need a refinement of the category $\VecBiLin$, which we are now going to motivate
by studying properties of the functor $\PhaseSpace:\GlobHypAffGreen\to\VecBiLin$, and showing that
it can be actually regarded as a functor to a proper subcategory of $\VecBiLin$.
\begin{lem}\label{lem:nullspacemorph}
Let $(f,\underline{f})$ be a morphism between two objects $(M_1,\bfA_1,P_1)$ and
$(M_2,\bfA_2,P_2)$ in $\GlobHypAffGreen$. We use the notation
$\PhaseSpace(M_i,\bfA_i,P_i) =: (\EE_i,\tau_i)$, $i=1,2$, for the phase spaces and
$\mathcal{N}_i \subseteq \EE_i$,  $i=1,2$, for the null spaces. 
It holds true that $\PhaseSpace(f,\underline{f})\big[\mathcal{N}_1\big] \subseteq \mathcal{N}_2$.
\end{lem}
\begin{proof}
From Proposition \ref{propo:degen} we know that any element in $\mathcal{N}_1$ is of
the form $[a\,\1_1]$, with $a\in C^\infty_0(M_1)$.
Using the explicit form of $\PhaseSpace(f,\underline{f})$ given in (\ref{eqn:sympmorphism}) we obtain
\begin{flalign}
\PhaseSpace(f,\underline{f})\big([a\,\1_1]\big) = [{f^\dagger}_\ast(a\,\1_1)] = [\underline{f}_\ast(a)\,\1_2]~.
\end{flalign}
Then because of Proposition \ref{propo:degen}  we have $\PhaseSpace(f,\underline{f})\big([a\,\1_1]\big)\in\mathcal{N}_2$.
\end{proof}
Notice that due to Lemma \ref{lem:nullspacemorph} any morphism $\PhaseSpace(f,\underline{f}):\EE_1\to \EE_2$ restricts
to the quotients $\EE_i /\mathcal{N}_i$, $i=1,2$. Motivated by this property and
Proposition \ref{propo:degen} we refine the category $\VecBiLin$ as follows:
\begin{defi}\label{defi:refinedvec}
The category $\VecBiLinR$ consists of the following objects and morphisms:
\begin{itemize}
\item An object in $\VecBiLinR$ is a tuple $(\EE,\tau)$, where
$\EE$ is a (possibly infinite-dimensional) vector space and $\tau:\EE\times\EE\to \bbR$ is a bilinear map, such that
the left and right null spaces of $\tau$ coincide, i.e.~$\mathcal{L}=\mathcal{R}=:\mathcal{N}$.
\item A morphism between two objects $(\EE_1,\tau_1)$ and $(\EE_2,\tau_2)$ in $\VecBiLinR$
is an injective linear map $L:\EE_1\to\EE_2$, which preserves the bilinear maps, i.e.~$\tau_{2}(L(v),L(w))=\tau_1(v,w)$, for
all $v,w\in \EE_1$, and also preserves the null spaces, i.e.~$L\big[\mathcal{N}_1\big]\subseteq \mathcal{N}_2$.
\end{itemize}
\end{defi}
\sk

The linear phase space functor $\PhaseSpace_\lin$ is a functor from 
$\mathsf{GlobHypLinGreen}$ to the subcategory $\VecBiLinR$ of $\VecBiLin$, since the null spaces are trivial. We also find that this
holds true for the functor $\PhaseSpace$.
\begin{cor}
The functor $\PhaseSpace$ of Theorem  \ref{theo:PhaseSpace} is a functor from
$\GlobHypAffGreen$ to the subcategory $ \VecBiLinR$ of $\VecBiLin$.
\end{cor}
\begin{proof}
Let $(M,\bfA,P)$ be an object in $\GlobHypAffGreen$. Due to Proposition \ref{propo:degen} we have that
$\mathcal{L} = \mathcal{R} =\mathcal{N}$ for the object $\PhaseSpace(M,\bfA,P)$ in $\VecBiLin$. Hence, 
$\PhaseSpace(M,\bfA,P)$ is an object in $\VecBiLinR$.

Let now $(f,\underline{f})$ be a morphism between two objects $(M_1,\bfA_1,P_1)$ and $(M_2,\bfA_2,P_2)$
in $\GlobHypAffGreen$. Due to Lemma \ref{lem:nullspacemorph} we have that
the morphism $\PhaseSpace(f,\underline{f})$ in $\VecBiLin$ preserves the null spaces.
Hence, $\PhaseSpace(M,\bfA,P)$ is a morphism in $\VecBiLinR$.
\end{proof}
We can now make precise the lower horizontal arrow in the heuristic diagram (\ref{eqn:wishlist}).
\begin{propo}
There is a covariant functor $\mathfrak{LinPhSp}:\VecBiLinR\to\VecBiLinR$.
It associates to any object $(\EE,\tau)$ in $\VecBiLinR$ the object $(\EE/\mathcal{N},\tau)$ in
$\VecBiLinR$ and to any morphism $L:\EE_1\to\EE_2$ between two objects $(\EE_i,\tau_i)$, $i=1,2$, in
$\VecBiLinR$ the canonically induced morphism $L:\EE_1/\mathcal{N}_1\to \EE_2/\mathcal{N}_2$.
\end{propo}
\begin{proof}
Notice that the bilinear map $\tau:\EE\times\EE\to\bbR$ induces canonically a bilinear
map (denoted by the same symbol) $\tau:\EE/\mathcal{N} \times \EE/\mathcal{N}\to \bbR$, hence
$(\EE/\mathcal{N},\tau)$ is an object in $\VecBiLin$. Since the null spaces are trivial it is also an object 
in $\VecBiLinR$.

Let $L:\EE_1\to \EE_2$ be a morphism between two objects $(\EE_i,\tau_i)$, $i=1,2$, in
$\VecBiLinR$. Since $L[\mathcal{N}_1]\subseteq \mathcal{N}_2$, the map $L$ canonically induces
a linear map (denoted by the same symbol) $L:\EE_1/\mathcal{N}_1\to\EE_2/\mathcal{N}_2$.
This induced map is injective: Let $v\in \EE_1$ be such that $L(v) \in \mathcal{N}_2$.
Then, for all $w\in \EE_1$, $\tau_1(v,w) = \tau_2(L(v),L(w)) =0$ and thus $v\in \mathcal{N}_1$.
This implies that $L:\EE_1/\mathcal{N}_1\to\EE_2/\mathcal{N}_2$ is injective and hence a morphism
in $\VecBiLin$. Since the null spaces of $\EE_i/\mathcal{N}_i$, $i=1,2$, are trivial
the induced map is also a morphism in $\VecBiLinR$.
\end{proof}
We can now finally make precise the heuristic diagram (\ref{eqn:wishlist}).
\begin{theo}\label{theo:natiso}
There is a natural isomorphism between the two covariant functors
$\PhaseSpace_\lin\circ \mathfrak{LinBund}$ and $\mathfrak{LinPhSp}\circ \PhaseSpace$
from $\GlobHypAffGreen$ to $\VecBiLinR$.
\end{theo}
\begin{proof}
In the proof we use the simplified notations $\mathfrak{F}:= \PhaseSpace_\lin\circ \mathfrak{LinBund}$ 
 and $\mathfrak{G}:= \mathfrak{LinPhSp}\circ \PhaseSpace$ for the two functors.
 
Let $(M,\bfA,P)$ be an object in $\GlobHypAffGreen$. 
Applying $\mathfrak{F}$ we obtain the object $(\EE_\lin,\tau_\lin)$ in $\VecBiLinR$, where
$\EE_\lin : = \sectn{M}{\bfV}/P_\bfV\big[\sectn{M}{\bfV}\big]$ and $\tau_\lin$ is given by
\begin{flalign}
\tau_\lin:\EE_\lin\times \EE_\lin \to \bbR~,~~([h],[k])\mapsto \tau_\lin([h],[k]) = \int_M\vol~\ip{h}{G_\bfV(k)}~.
\end{flalign}
Applying $\mathfrak{G}$ we obtain the object $(\EE/\mathcal{N},\tau)$ in $\VecBiLinR$, where (cf.~Proposition \ref{propo:sympl})
$\EE := \big(\sectn{M}{\bfA^\dagger}/\mathrm{Triv}\big)/P^\ast\big[\sectn{M}{\bfV}\big]$, $\mathcal{N}$ is the null space
characterized in Corollary \ref{cor:nullspace} and $\tau$ is given in (\ref{eqn:tau}).

Let us fix any section $\widehat{s}\in \sect{M}{\bfA}$ and define a linear map
\begin{flalign}\label{eqn:naturaliso}
\eta_{(M,\bfA,P)}: \sectn{M}{\bfV}\mapsto \EE/\mathcal{N}~,~~h\mapsto \big[\ip{h}{\cdot - \widehat{s}}\big]~.
\end{flalign}
This map does not depend on the choice of $\widehat{s}$. Indeed, let $\widetilde{s}\in \sect{M}{\bfA}$ be another section,
then $\widehat{s} = \widetilde{s} +\sigma$ for some unique $\sigma\in\sect{M}{\bfV}$ and we obtain, for all $h\in \sectn{M}{\bfV}$,
\begin{flalign}
\big[\ip{h}{\cdot - \widehat{s}}\big] = \big[\ip{h}{\cdot - \widetilde{s} - \sigma}\big] 
= \big[\ip{h}{\cdot - \widetilde{s}} - \ip{h}{\sigma} \,\1\big] =\big[\ip{h}{\cdot - \widetilde{s}}\big]~,
\end{flalign}
where in the last equality we have used the characterization of $\mathcal{N}$ (cf.~Corollary \ref{cor:nullspace}).

The linear map (\ref{eqn:naturaliso}) canonically induces a linear map (denoted by the same symbol)
$\eta_{(M,\bfA,P)} : \EE_\lin \to \EE/\mathcal{N}$, since, for all $h\in \sectn{M}{\bfV}$, 
$\big[\ip{P_\bfV(h)}{\cdot - \widehat{s}}\big] =0$. To prove this last statement
remember the explicit form of $P^\ast$ given in (\ref{eqn:temp}). Using again the characterization of $\mathcal{N}$
 (cf.~Corollary \ref{cor:nullspace}) we obtain, for all $h\in \sectn{M}{\bfV}$,
 \begin{flalign}
 \big[\ip{P_\bfV(h)}{\cdot - \widehat{s}}\big] = \big[ \ip{h}{P(\widehat{s})}\,\1 + \ip{P_\bfV(h)}{\cdot - \widehat{s}} \big] = \big[P^\ast(h)\big] = 0~.
 \end{flalign}
The map $\eta_{(M,\bfA,P)}$ is a morphism in $\VecBiLinR$, since it preserves the bilinear maps, for 
all $[h],[k]\in\EE_\lin$,
\begin{flalign}
\tau\big(\eta_{(M,\bfA,P)}([h]),\eta_{(M,\bfA,P)}([k])\big)= \int_M\vol~\ip{h}{G_\bfV(k)} =\tau_\lin([h],[k])~.
\end{flalign}
 
 We next construct the inverse of this map. Let us consider the linear map
  $\sectn{M}{\bfA^\dagger} \to \sectn{M}{\bfV}\,,~\varphi \mapsto \varphi_\bfV$, which takes the linear part, and induce
  the linear map
  \begin{flalign}
  \eta_{(M,\bfA,P)}^{-1} : \sectn{M}{\bfA^\dagger} \to \EE_\lin~,~~\varphi \mapsto [\varphi_\bfV]~.
  \end{flalign}
  This map induces a well-defined linear map on the quotient $\EE/\mathcal{N}$ (which we denote with
  a slight abuse of notation by the same symbol), $ \eta_{(M,\bfA,P)}^{-1}: \EE/\mathcal{N}\to \EE_\lin$.
  This is the inverse of $\eta_{(M,\bfA,P)}$, since, for all $[h]\in\EE_\lin$,
\begin{subequations}  
  \begin{flalign}
  \eta_{(M,\bfA,P)}^{-1}\circ \eta_{(M,\bfA,P)}\big([h]\big)
  = \eta_{(M,\bfA,P)}^{-1}\big(\big[\ip{h}{\cdot -\widehat{s}}\big]\big) = [h]~,
  \end{flalign}
  and, for all $[\varphi]\in \EE/\mathcal{N}$,
  \begin{flalign}
 \nn \eta_{(M,\bfA,P)}\circ \eta_{(M,\bfA,P)}^{-1}\big([\varphi]\big) &= \eta_{(M,\bfA,P)}\big([\varphi_\bfV]\big) = \big[\ip{\varphi_\bfV}{\cdot -\widehat{s}}\big]\\
  &=\big[\varphi(\widehat{s})\,\1 + \ip{\varphi_\bfV}{\cdot -\widehat{s}}\big] = [\varphi]~.
  \end{flalign}
 \end{subequations}
 The family of morphisms $\eta_{(M,\bfA,P)}$ labelled by objects in $\GlobHypAffGreen$ provides a natural
 isomorphism, since for all morphisms $(f,\underline{f})$ between
 two objects $(M_i,\bfA_i,P_i)$, $i=1,2$, in $\GlobHypAffGreen$ and for all $\varphi\in\sectn{M_1}{\bfA_1^\dagger}$,
we have ${f^\dagger}_\ast(\varphi)_\bfV ={ f_\bfV}_\ast(\varphi_\bfV)$. Hence, the following diagram commutes
 \begin{flalign}
 \xymatrix{
 \mathfrak{F}(M_1,\bfA_1,P_1) \ar[d]_-{\eta_{(M_1,\bfA_1,P_1)}}\ar[rr]^-{\mathfrak{F}(f,\underline{f})} & & \mathfrak{F}(M_2,\bfA_2,P_2) \ar[d]^-{\eta_{(M_2,\bfA_2,P_2)}}\\
 \mathfrak{G}(M_1,\bfA_1,P_1)  \ar[rr]^-{\mathfrak{G}(f,\underline{f})} & & \mathfrak{G}(M_2,\bfA_2,P_2)
 }
 \end{flalign}
\end{proof}
For bosonic or fermionic theories we also obtain a natural isomorphism between functors from
the corresponding subcategories of $\GlobHypAffGreen$ to $\astAlg$.
\begin{cor}\label{cor:linalgebra}
There is a natural isomorphism between the two covariant functors
$ \CCR \circ\PhaseSpace_\lin\circ \mathfrak{LinBund}$ and $\CCR \circ \mathfrak{LinPhSp}\circ \PhaseSpace$
from $\GlobHypAffGreen^\mathrm{bos}$ to $\astAlg$.

Similarly, there is a natural isomorphism between the two covariant functors
$ \CAR \circ\PhaseSpace_\lin\circ \mathfrak{LinBund}$ and $\CAR \circ \mathfrak{LinPhSp}\circ \PhaseSpace$
from $\GlobHypAffGreen^\mathrm{ferm}$ to $\astAlg$.
\end{cor}
\begin{proof}
The natural isomorphism constructed in Theorem \ref{theo:natiso}
canonically lifts to natural isomorphisms (denoted by the same symbols) 
$\eta_{(M,\bfA,P)}:\CCR(\EE_\lin,\tau_\lin) \to \CCR(\EE/\mathcal{N},\tau)$ (for bosonic theories)
and $\eta_{(M,\bfA,P)}:\CAR(\EE_\lin,\tau_\lin) \to \CAR(\EE/\mathcal{N},\tau)$ (for fermionic theories).
\end{proof}
%%%%%%%%%%%%%%%%%%%%%%%%%%%%%%%%%%%%%%%%%%%%%%%%

\section{\label{sec:states}Induction of states}
In this section we prove that physically relevant states of affine quantum field theories can be
induced from those of their linear counterparts.
In particular we are interested in states satisfying the Hadamard condition.
In \cite{Brunetti:2001dx} it has been shown that the assignment of a state space to a quantum field theory 
can be described in terms of a suitable contravariant functor, however, the identification of a preferred
Hadamard state on all spacetimes is not possible.
This argument has been strengthened by a no-go theorem proven by Fewster and Verch \cite{Fewster:2011pe},
which holds under the assumption of the dynamical locality property. 
 We expect that a similar result holds true also in the case under study. 
Hence, we shall discuss in this section only the construction of states on a fixed background geometry
and do not attempt to give a categorical description.
We furthermore spell out our results only for bosonic quantum field theories, since
the constructions for fermionic ones are analogous.
\sk

Given an object $(M,\bfA,P)$ in $\GlobHypAffGreen^\mathrm{bos}$ we can construct a unital
$\ast$-algebra $\QFT^\mathrm{bos}(M,\bfA,P)$ by applying the functor $\QFT^\mathrm{bos} :=
 \CCR\circ \PhaseSpace: \GlobHypAffGreen^\mathrm{bos} \to \astAlg$.
We are interested in constructing states on this algebra which satisfy the Hadamard condition. 
This condition has been first imposed for linear theories, since it implies that
the ultraviolet behavior of the state coincides with that of the Minkowski vacuum 
and that the quantum fluctuations of all observables are bounded. The mathematical 
relevance of the Hadamard condition becomes manifest when one includes perturbative interactions
or studies an extended class of observables, containing e.g.~the stress energy tensor. In this
 case the algebra of field polynomials has to be enlarged to include also the so-called Wick polynomials and a locally 
 covariant notion of Wick products can be defined by exploiting Hadamard states \cite{Hollands:2001fb}.
\sk

Notice that we have used in the text above a slight abuse of language when referring to ``Hadamard condition''.
Properly speaking, an algebraic state is physical if it satisfies the 
{\em microlocal spectrum condition ($\mu$SC)} \cite{Brunetti:1995rf}, 
which is a stronger requirement and equivalent to the conventional Hadamard condition only under certain assumptions
(see \cite{Sanders:2009sw}). 
\sk

Since there has been a considerable amount of work on the explicit construction of states of Hadamard type
 for linear quantum field theories \cite{Dappiaggi:2005ci, Dappiaggi:2007mx, Dappiaggi:2010gt, Hack:2010iw, Dappiaggi:2011cj},
 we in particular can ask whether there exists a way to induce physically relevant states 
 from linear to affine quantum field theories. Indeed, using the linearization functor discussed in Section \ref{sec:linfunctor}
we can associate to $(M,\bfA,P)$ also the unital $\ast$-algebra $\QFT^\mathrm{bos}_\lin(M,\bfA,P)$
 of the linearized theory by applying the functor
$\QFT^\mathrm{bos}_\lin
 := \CCR \circ \PhaseSpace_\lin\circ \mathfrak{LinBund}:\GlobHypAffGreen^\mathrm{bos} \to \astAlg$.
Let us consider a state $\Omega$ for the linear quantum field theory, that is a positive and normalized linear functional 
from $\QFT^\mathrm{bos}_\lin (M,\bfA,P)$ to the complex numbers. 
We shall also assume that $\Omega$ is quasi-free and satisfies the Hadamard condition (which in this case is
equivalent to the $\mu$SC  \cite{Sanders:2009sw}).  
Any unital $\ast$-algebra homomorphism $\kappa: \QFT^\mathrm{bos}(M,\bfA,P)\to \QFT^\mathrm{bos}_\lin (M,\bfA,P)$
allows us to induce a state $\Omega_\kappa$ on $\QFT^\mathrm{bos}(M,\bfA,P)$ via the pull-back construction,
 for all $b\in \QFT^\mathrm{bos} (M,\bfA,P)$,
\begin{flalign}
\Omega_\kappa(b) := \Omega\big(\kappa(b)\big)~.
\end{flalign}
We are now proposing a particularly simple class of such unital $\ast$-algebra homomorphisms.
The pull-back of a quasi-free and Hadamard state $\Omega$ on $\QFT^\mathrm{bos}_\lin (M,\bfA,P)$ 
is investigated in detail and it is found that this state is not quasi-free (this is actually what one wants in an affine theory, see Remark
\ref{rem:state}), but it still satisfies the microlocal spectrum condition.
\sk

Let us consider the two vector spaces $\EE= \big(\sectn{M}{\bfA^\dagger}/\mathrm{Triv}\big)/P^\ast\big[\sectn{M}{\bfV}\big]$
and  $\EE_\lin \oplus \bbR$, where $\EE_\lin = \sectn{M}{\bfV}/P_\bfV\big[\sectn{M}{\bfV}\big]$.
Given any $\widehat{s}\in\sect{M}{\bfA}$ satisfying the affine equation of motion $P(\widehat{s}) =0$ there is a well-defined
linear map
\begin{flalign}
\EE \to \EE_\lin \oplus \bbR~,~~[\varphi] \mapsto [\varphi_\bfV] + \int_M\vol~\varphi(\widehat{s})~,
\end{flalign}
which lifts to a unital $\ast$-algebra homomorphism $\kappa_{\widehat{s}}: \QFT^\mathrm{bos}(M,\bfA,P)
\to \QFT_\lin^\mathrm{bos}(M,\bfA,P)$.
Indeed, defining on the generators $\Psi([\varphi])$, $[\varphi]\in \EE$, of $\QFT^\mathrm{bos}(M,\bfA,P)$ the map by
 $\kappa_{\widehat{s}}\big(\Psi([\varphi])\big) :=  \Psi_\lin\big([\varphi_\bfV]\big) +\int_M\vol~\varphi(\widehat{s})~~\oone $, 
for all $[\varphi]\in \EE$, canonically induces a unital $\ast$-algebra homomorphism. Here
 $\Psi_\lin\big([h]\big)$, $[h]\in\EE_\lin$, are the generators of
  $\QFT_\lin^\mathrm{bos}(M,\bfA,P)$ and $\oone$ the unit in this algebra.
Notice that the pull-back state $\Omega_{\kappa_{\widehat{s}}}$ on $\QFT^\mathrm{bos}(M,\bfA,P)$
is not quasi-free. In particular, the one-point distribution reads, for all $[\varphi]\in\EE$,
\begin{flalign}
 \Omega_{\kappa_{\widehat{s}}}\big(\Psi([\varphi])\big) = \Omega\Big(\kappa_{\widehat{s}}\big(\Psi([\varphi])\big)\Big)
= \Omega\Big(\Psi_\lin([\varphi_\bfV]) + \int_M\vol~\varphi(\widehat{s}) ~~\oone\Big)
=\int_M\vol~\varphi(\widehat{s}) ~.\label{eqn:onepoint}
\end{flalign}
For the two-point distribution we obtain, for all $[\varphi],[\psi]\in \EE$,
\begin{flalign}\label{eqn:twopoint}
\Omega_{\kappa_{\widehat{s}}}\big(\Psi([\varphi])\,\Psi([\psi])\big) = \Omega\big(\Psi_\lin([\varphi_\bfV])\,\Psi_\lin([\psi_\bfV])\big)
+ \left(\int_M\vol~\varphi(\widehat{s})\right)~\left(\int_M\vol~\psi(\widehat{s})\right)~.
\end{flalign}
\begin{rem}\label{rem:state}
The physical interpretation of the unital $\ast$-algebra homomorphism $\kappa_{\widehat{s}}$ is as follows:
We use some {\it fixed} reference solution $\widehat{s}$ in order to split any section $\varphi\in\sectn{M}{\bfA^\dagger}$
as $\varphi = \varphi(\widehat{s}) \,\1 + \ip{\varphi_\bfV}{\cdot -\widehat{s}}$.
Due to Proposition \ref{propo:degen} the first term is an element in the null space of $\tau$, hence in the quantum theory
an element in the center of the algebra. 
The classical observable (\ref{eqn:laf}) associated to the first term can be interpreted to measure 
the background solution $\widehat{s}$, and so we should interpret the corresponding quantum observable.
The second term is interpreted to measure the fluctuations
around the background solution  $\widehat{s}$. The state $\Omega_{\kappa_{\widehat{s}}}$ we have constructed
has a one-point distribution  (\ref{eqn:onepoint}) which reproduces classical results (i.e.~the correspondence principle
holds true) and a two-point distribution (\ref{eqn:twopoint}) which becomes  (upon truncation)  
the two-point function of the linear quantum field theory describing the
 fluctuations around $\widehat{s}$.
\end{rem}
\sk

We now study in detail the properties of the state $\Omega_{\kappa_{\widehat{s}}}$,
following in particular the idea and results of \cite{Sanders:2009sw}.
Let $\widetilde\Omega:\QFT^\mathrm{bos}(M,\bfA,P)\to\mathbb{C}$ be an arbitrary state and let us denote by
 $\widetilde{\omega}_n$, $n\in \bbN$, the $n$-point distribution of $\widetilde\Omega$, i.e.~the
 distribution $\widetilde{\omega}_n:\sectn{M}{\bfA^\dagger}^{\times n}\to\mathbb{C}$ defined by, for all
 $\varphi_1,\dots,\varphi_n\in \sectn{M}{\bfA^\dagger}$,
\begin{flalign}\label{eqn:npoint}
\widetilde{\omega}_n\big(\varphi_1,\dots,\varphi_n\big) :=\widetilde\Omega\big(\Psi([\varphi_1])\dots \Psi([\varphi_n])\big)~.
\end{flalign}
We denote the integral kernel of this distribution  by $\widetilde{\omega}_n(x_1,\dots,x_n)$. 
For all $n\geq 1$, let $\mathcal{I}_n$ be the set of all partitions of $\{1,\dots ,n\}$.
For all $I\in\mathcal{I}_n$, any element $i\in I$ is an ordered set with $\vert i\vert$ elements, namely
 $i(1)< \dots <i(\vert i\vert)$. The {\em truncated $n$-point distribution} $\widetilde\omega_n^{\mathrm{T}}(x_1,...,x_n)$, 
 $n\in \bbN$, of  $\widetilde\Omega$ is then defined implicitly by
\begin{flalign}\label{trunc}
\widetilde{\omega}_n(x_1,\dots ,x_n)=
\sum\limits_{I\in\mathcal{I}_n}\prod\limits_{i\in I}\widetilde{\omega}^{\mathrm{T}}_{\vert i\vert}(x_{i(1)},\dots ,x_{i(\vert i\vert)})~.
\end{flalign}

\begin{defi}\label{defi:affqf}
A state $\widetilde\Omega:\QFT^\mathrm{bos}(M,\bfA,P)\to\mathbb{C}$ is called 
{\bf affine quasi-free} if and only if $\widetilde\omega^{\mathrm{T}}_n=0$, for all $n>2$.
\end{defi}

In Definition \ref{defi:affqf} we deviate from the standard definition of quasi-free states, where it is also assumed
that the one-point distribution vanishes. This generalization is motivated by Remark \ref{rem:state}.
\begin{propo}\label{affinequasifree}
Let $\Omega:\QFT^\mathrm{bos}_\lin(M,\bfA,P)\to\mathbb{C}$ be a quasi-free state.
Then the state $\Omega_{\kappa_{\widehat{s}}}:\QFT^\mathrm{bos}(M,\bfA,P)\to\mathbb{C}$ 
is affine quasi-free.
\end{propo}
\begin{proof}
We denote by $\omega_{\kappa_{\widehat{s}} , n}$, $n\in\bbN$, the $n$-point distribution
of the state $\Omega_{\kappa_{\widehat{s}}}:\QFT^\mathrm{bos}(M,\bfA,P)\to\mathbb{C}$
and by $\omega_{\bfV , n}$, $n\in\bbN$, the $n$-point distribution of
the state $\Omega_{\bfV}:\QFT^\mathrm{bos}(M,\bfA,P)\to\mathbb{C}$ defined by,
for all $[\varphi_1],\dots,[\varphi_n]\in \EE$,
\begin{flalign}\label{eqn:statetempV}
\Omega_\bfV\big(\Psi([\varphi_1])\dots \Psi([\varphi_n])\big) := \Omega\big(\Psi_\lin([\varphi_{1\,\bfV}])
\dots \Psi_\lin([\varphi_{n\,\bfV}])\big)~.
\end{flalign}
The hypothesis that $\Omega$ is quasi-free implies that $\Omega_\bfV$ is also quasi-free, 
i.e.~$\omega^\mathrm{T}_{\bfV , n} =0$, for all $n\neq 2$, and $\omega_{\bfV,2} = \omega_{\bfV , 2}^\mathrm{T}$.
From (\ref{eqn:onepoint}) and (\ref{eqn:twopoint}) we obtain 
$\omega_{\kappa_{\widehat{s}} , 1}(x) = \omega^\mathrm{T}_{\kappa_{\widehat{s}} , 1}(x) = \mathrm{ev}_{\widehat{s}}(x)$
 (where $\mathrm{ev}$ denotes the evaluation map) and $\omega_{\kappa_{\widehat{s}} ,  2}(x_1,x_2) = 
 \omega_{\bfV ,  2}(x_1,x_2)  + \omega_{\kappa_{\widehat{s}} ,  1}(x_1)~ \omega_{\kappa_{\widehat{s}} ,  1}(x_2)$, thus
 $\omega^{\mathrm{T}}_{\kappa_{\widehat{s}} ,  2}(x_1,x_2) = \omega_{\bfV,  2}(x_1,x_2)$.
 A direct calculation similar to (\ref{eqn:onepoint}) and (\ref{eqn:twopoint}) shows 
 that $\omega_{\kappa_{\widehat{s}} , 3}^{\mathrm{T}} = 0$. We now show by induction that 
 $\omega_{\kappa_{\widehat{s}} , n}^\mathrm{T} = 0$, for all $n >2$, and hence that 
 $\Omega_{\kappa_{\widehat{s}}}$ is affine quasi-free: Assume that $\omega_{\kappa_{\widehat{s}} , m}^\mathrm{T} = 0$,
 for all $2<m <n$, then for $n$ the right hand side of (\ref{trunc}) simplifies to
 \begin{flalign}
\nn \sum\limits_{I\in\mathcal{I}_n}\prod\limits_{i\in I} &\omega_{\kappa_{\widehat{s}} , \vert i\vert}^{\mathrm{T}}(x_{i(1)},\dots ,x_{i(\vert i\vert)}) \\
\nn &= \omega_{\kappa_{\widehat{s}} , n}^{\mathrm{T}}(x_1,\dots ,x_n)   +
  \sum\limits_{I\in\mathcal{I}_n\vert^{}_{\mathrm{2,1}}}\prod\limits_{i\in I_2} \omega_{\kappa_{\widehat{s}} , 2}^{\mathrm{T}}(x_{i(1)},x_{i(2)})~\prod\limits_{j\in I_1}\omega_{\kappa_{\widehat{s}} , 1}^{\mathrm{T}}(x_{j(1)})\\
 &= \omega_{\kappa_{\widehat{s}} , n}^{\mathrm{T}}(x_1,\dots ,x_n)   +
  \sum\limits_{I\in\mathcal{I}_n\vert^{}_{\mathrm{2,1}}}\prod\limits_{i\in I_2} \omega_{\bfV , 2}(x_{i(1)},x_{i(2)})~\prod\limits_{j\in I_1}\omega_{\kappa_{\widehat{s}} , 1}(x_{j(1)})~,\label{eqn:stateqftemp}
 \end{flalign}
 where we have denoted by $\mathcal{I}_n\vert^{}_{\mathrm{2,1}}$ the set of all partitions of $\{1,\dots,n\}$
 into pairs and singlets and  by $I = :I_1\cup I_2$ the split into singlets and pairs.
 Since $\Omega_\bfV$ is quasi-free, its $m$-point distribution factorizes, for all $m>1$, in two-point distributions
 \begin{flalign}
 \omega_{\bfV , m}(x_1,\dots, x_m) = \sum\limits_{I\in \mathcal{I}_m\vert^{}_{2}} \prod\limits_{i\in I} \omega_{\bfV , 2}(x_{i(1)},x_{i(2)})~,
\end{flalign}
where we have denoted by $\mathcal{I}_m\vert^{}_{\mathrm{2}}$ the set of all partitions of $\{1,\dots,m\}$
 into pairs (in our conventions the sum over the empty set $\mathcal{I}_m\vert^{}_{\mathrm{2}}$, for $m$ odd, is zero).
Using  this expression in order to simplify the second term in (\ref{eqn:stateqftemp}) yields
  \begin{flalign}
\nn \sum\limits_{I\in\mathcal{I}_n}\prod\limits_{i\in I} &\omega_{\kappa_{\widehat{s}} , \vert i\vert}^{\mathrm{T}}(x_{i(1)}, \dots ,x_{i(\vert i\vert)}) \\
\nn &= \omega_{\kappa_{\widehat{s}} , n}^{\mathrm{T}}(x_1,\dots ,x_n)   +\sum\limits_{k=0}^n
  \sum\limits_{I\in\mathcal{I}_n\vert^{}_{\mathrm{(n-k)\ast,1}}}\omega_{\bfV , n-k}(x_{i(1)},\dots, x_{i(n-k)})~\prod\limits_{j\in I_1}\omega_{\kappa_{\widehat{s}} , 1}(x_{j(1)})\\
  &= \omega_{\kappa_{\widehat{s}} , n}^{\mathrm{T}}(x_1,\dots ,x_n)    + \omega_{\kappa_{\widehat{s}} , n}(x_1,\dots ,x_n)   ~,
  \end{flalign}
  where we have denoted by $\mathcal{I}_n\vert^{}_{\mathrm{(n-k)\ast,1}}$ the set of all partitions of $\{1,\dots,n\}$
  into {\it exactly one} subset with $n-k$ elements and consequently $k$ singlets. The last equality follows by considerations
    similarly to  (\ref{eqn:twopoint}). Hence, (\ref{trunc}) implies that $ \omega_{\kappa_{\widehat{s}} , n}^{\mathrm{T}}=0$.
\end{proof}

We now prove that the state $\Omega_{\kappa_{\widehat{s}}}$ is physical.
In contrast to what is common in linear quantum field theories with standard quasi-free states,
this proof does not only consist of verifying the Hadamard condition, but rather the full microlocal spectrum condition.
We review its definition starting with the following ancillary structure \cite{Sanders:2009sw}:
\begin{defi}
Let $\mathcal{G}_n$, $n\in \bbN$, be the set of all graphs with $n$ vertices and finitely many edges.
 An immersion of a graph $G\in\mathcal{G}_n$ into a spacetime $M$ consists of the following assignments:
 \begin{itemize}
  \item[(i)] a point $x(i)\in M$ to each vertex $\nu_i$ of $G$, 
  \item[(ii)] a piecewise smooth curve $\gamma_r(i,j)$ between 
  $x(i)$ and $x(j)$ to every edge $e_r$ connecting the vertices $\nu_i$ and $\nu_j$ in $G$,
  \item[(iii)] a causal, future pointing 
  and covariantly constant covector field $k_r$ on each $\gamma_r(i,j)$ to each $e_r$.
  \end{itemize}
   Let $\mathcal{Z}$ denote the $0$-section in $T^\ast M^n$. We say that a point 
   $(x_1,\zeta_1;\dots ;x_n,\zeta_n)\in T^\ast M^n\setminus\mathcal{Z}$
   is instantiated by an immersion of a graph $G\in\mathcal{G}_n$ if and only if for each $i=1,\dots,n$ the 
   immersion sends the vertex $\nu_i$ to $x_i$ and 
   \begin{flalign}
   \zeta_i=\sum\limits_{\substack{\gamma_r(i,j)\\ i<j}}k_r(x_i)-\sum\limits_{\substack{\gamma_r(i,j)\\ i>j}}k_r(x_i)~.
   \end{flalign}
\end{defi}
\sk

The singularities allowed for the $n$-point distribution of a physical state can be described in terms of the set
\begin{flalign}
\nn \Gamma_n:= \{ &(x_1,\zeta_1;\dots;x_n,\zeta_n)\in T^\ast M^n\setminus\mathcal{Z}\;:\;\exists G\in\mathcal{G}_n\,\text{and an immersion of $G$ into}\\ 
&\text{$M$ which instantiates the point}\;(x_1,\zeta_1;\dots;x_n,\zeta_n) \}~.
\end{flalign}

\begin{defi}
Let $\widetilde{\Omega}:\QFT^\mathrm{bos}(M,\bfA,P)\to\mathbb{C}$ be any state and let $\mathrm{WF}(\widetilde{\omega}_n)$,
$n\in \bbN$,
 be the wavefront set of the $n$-point distribution. We say that $\widetilde{\Omega}$ satisfies the 
 {\bf microlocal spectrum condition} with smooth immersions if and only if, 
 for all $n\in\mathbb{N}$, $\mathrm{WF}(\widetilde{\omega}_n)\subset\Gamma_n$. 
\end{defi}
\sk

Notice that here we encounter a difference to the work of Sanders \cite{Sanders:2009sw}.
Our theories are not scalar fields, but fields where the vector dual bundle $\bfA^\dagger$ is an (in general nontrivial)
vector bundle of rank higher than one. The wavefront set of a (vector valued) distribution can be simply 
defined as the union of the wavefront sets of the component distributions obtained via
 a local basis expansion. This turns out to be an invariant concept due to the properties of 
 scalar wavefront sets \cite{SahlmannVerch}. The results of \cite{Sanders:2009sw} 
 can be translated to our setting. Of particular importance is Lemma 2.6 which, adapted to our 
 scenario, guarantees the following:
\begin{lem}
A state $\widetilde{\Omega}:\QFT^\mathrm{bos}(M,\bfA,P)\to\mathbb{C}$ satisfies the microlocal 
spectrum condition with smooth immersions if and only if $\mathrm{WF}(\widetilde{\omega}^{\mathrm{T}}_n)\subset\Gamma_n$,
for all $n\in\mathbb{N}$. Here $\widetilde{\omega}^{\mathrm{T}}_n$ stands for the truncated $n$-point distribution.
\end{lem} 

We can now prove the main result of this section:
\begin{propo}
Let $\Omega:\QFT^\mathrm{bos}_\lin(M,\bfA,P)\to\mathbb{C}$ be a quasi-free state satisfying the Hadamard condition.
Then the state $\Omega_{\kappa_{\widehat{s}}}:\QFT^\mathrm{bos}(M,\bfA,P)\to\mathbb{C}$
 satisfies the microlocal spectrum condition with smooth immersions.
\end{propo}
\begin{proof}
Proposition \ref{affinequasifree} guarantees that $\Omega_{\kappa_{\widehat{s}}}$ is an affine quasi-free state, 
i.e.~$\omega^{\mathrm{T}}_{\kappa_{\widehat{s}},n}=0$, for all $n>2$. 
Thus,  $\mathrm{WF}(\omega^{\mathrm{T}}_{\kappa_{\widehat{s}},n})=\emptyset\subset \Gamma_n$, for all $n>2$. 
 In the same proposition we have shown that $\omega^{\mathrm{T}}_{\kappa_{\widehat{s}},2}=
 \omega_{\bfV,2}$. Since the quasi-free state $\Omega$ is per hypothesis of Hadamard form,
 it holds that $\mathrm{WF}(\omega^\mathrm{T}_{\kappa_{\widehat{s}},2})=\mathrm{WF}(\omega_{\bfV,2})
 =\mathrm{WF}(\omega_{2})\subset\Gamma_2$, where $\omega_2$ denotes the two-point distribution of $\Omega$.
 The second equality follows from the definition of the state $\Omega_\bfV$ (\ref{eqn:statetempV}).
  To conclude, we notice that $\omega^{\mathrm{T}}_{\kappa_{\widehat{s}},1}=\mathrm{ev}_{\widehat{s}}
  \in\Gamma^\infty(M,{\bfA^\dagger}^\ast)$, with ${\bfA^\dagger}^\ast$ denoting the dual of the vector dual bundle
  and $\mathrm{ev}$ the evaluation map, and thus 
  $\mathrm{WF}(\omega^{\mathrm{T}}_{\kappa_{\widehat{s}},1})=\emptyset\subset\Gamma_1$.
\end{proof}
%%%%%%%%%%%%%%%%%%%%%%%%%%%%%%%%%%%%%%%%%%%%%%%%

\section{\label{sec:inhommat}Example: Inhomogeneous linear matter field theory}
To illustrate the interesting structures available in affine matter field theories we 
are going to discuss the simple model of Example \ref{ex:inhomlin} in more detail. 
Let us start with a usual {\it linear} matter field theory,
i.e.~an object $(M,\bfV,P_\bfV)$ in $\mathsf{GlobHypLinGreen}$.
To simplify the presentation we assume that the nondegenerate bilinear form $\ip{~}{~}$ is symmetric, i.e.~that the field theory is
bosonic. The fermionic case follows analogously.

The usual field equation one imposes for such a system is $P_\bfV(s) =0$, with
$s\in\sect{M}{\bfV}$. Let us now assume that for some physical reason there is an external source 
$J\in\sect{M}{\bfV}$, such that the dynamics is governed by the inhomogeneous field equation
$P(s) := P_\bfV(s) + J =0$. Since the operator $P$ is not a linear differential operator, this system can not
be described by the methods developed in  \cite{Bar:2011iu}. 

Remember that any vector bundle is canonically an affine bundle modeled on itself (cf.~Example \ref{ex:vectorbundasaffinebund}).
 Let us denote by 
$\bfA$ the bundle $\bfV$ when regarded as an affine bundle. Then $(M,\bfA,P)$ is an object in
$\GlobHypAffGreen^\mathrm{bos}$, since
\begin{flalign}
P:\sect{M}{\bfA}\to\sect{M}{\bfV}~,~~s\mapsto P(s) = P_\bfV(s) +  J~
\end{flalign}
is an affine Green-hyperbolic operator. The linear part of $P$ is $P_\bfV$.
Thus, inhomogeneous linear matter field theories can be interpreted as affine matter field theories.
A scenario, which falls partly in this scheme, is electromagnetism where the dynamical field is
a $1$-form on the underlying spacetime and the electric current is treated as an external source.
 Since there is the additional complication of the local gauge  symmetry, a direct application of the 
 methods developed in our paper is not possible. A thorough analysis of this 
case will appear in \cite{SDH}.

There are distinct ways how to address the quantization of inhomogeneous matter field theories:
First, one can apply the methods for the quantization of affine matter field theories developed in the sections
above. Second, one can {\it choose} some solution $\widehat{s}\in\sect{M}{\bfV}$ of the inhomogeneous
field equation, i.e.~$P(\widehat{s}) =0$, and describe fluctuations $\sigma\in\sect{M}{\bfV}$ around this
$\widehat{s}$ in terms of a linear matter field theory with field equation $P_\bfV(\sigma)=0 $.
Besides the fact that this choice of $\widehat{s}$ is quite unnatural, there are also interesting differences
in the structure of the algebra of observables depending whether we follow the first or second approach. The goal
of the remaining part of this section is to discuss these differences.

\begin{rem}
A potential third method\footnote{We are grateful to the referee for pointing us out this method.} consists of constructing an algebra of observables as the quotient of the tensor algebra, built out of $\sectn{M}{V}$, by a suitable ideal. This is generated by the canonical anti/commutation relations and by elements of the form $P_V(h)-\int_M\vol~\ip{h}{J}~\oone$ where $\oone$ is the identity of the algebra. The $*$-operation is complex conjugation. The advantage of this procedure is that the resulting algebra is isomorphic to the standard field algebra, that is the one where we set $J=0$. Hence every Hadamard state for the latter induces for the algebra of the inhomogeneous system a state, which fulfills the microlocal spectrum condition. The disadvantage is that this is an ad-hoc construction which cannot be applied to the generic scenario we consider in the previous sections. For this reason we shall not further pursue this method.
\end{rem}

Let us start by following the second method and linearize the inhomogeneous theory around some solution $\widehat{s}$.
Applying the quantization methods of \cite{Bar:2011iu} (see also Section \ref{sec:linfunctor}) 
one arrives at the unital $\ast$-algebra of observables $\QFT^\mathrm{bos}_\lin(M,\bfA,P)$,
whose generators can be interpreted as quantizations of the following classical linear functionals, for
$h\in\sectn{M}{\bfV}$,
\begin{flalign}\label{eqn:obslin}
F_h : \sect{M}{\bfV}\to \bbR~,~~\sigma \mapsto F_h(\sigma) = \int_M\vol~\ip{h}{\sigma}~,
\end{flalign}
modulo equation of motion. It is important to notice that from the classical observables
$F_h$ no information on the source term $J$ can be obtained. The same holds true for the quantum observables.
All the information on $J$ is absorbed in the section $\widehat{s}$, which does not enter the observable algebra.
However, any reasonable observable algebra should allow for observables determining 
the external source term. In this sense it is fair to say that the second approach to inhomogeneous linear matter
field theories does not lead to a complete description of those.

Let us now follow the first method and regard the inhomogeneous linear matter field theory as an affine matter
field theory. Applying the quantization methods developed in the sections above one arrives
at the unital $\ast$-algebra of observables $\QFT^\mathrm{bos}(M,\bfA,P)$, 
whose generators can be interpreted as quantizations of the following
classical affine functionals, for $\varphi\in\sectn{M}{\bfA^\dagger}$,
\begin{flalign}\label{eqn:obsaff}
F_\varphi: \sect{M}{\bfA} \to \bbR~,~~s \mapsto F_\varphi(s) = \int_M\vol~\varphi(s)~,
\end{flalign}
modulo trivial functionals and equation of motion. This set of classical observables in particular includes the observables
(\ref{eqn:obslin}) by choosing $\varphi = \ip{h}{\,\cdot\, -\widehat{s}\,}$, with $h\in\sectn{M}{\bfV}$.
Notice that here the role of $\widehat{s}$ is qualitatively different to the second approach. We employ an $\widehat{s}$
in order to specify observables which we can interpret to measure differences between $s$ and $\widehat{s}$, while
in the second approach a fixed choice of $\widehat{s}$ is necessary to construct the observable algebra.
Even more, the external source $J$ can be completely reconstructed from the classical observables (\ref{eqn:obsaff}):
Choose $\varphi = \ip{P_\bfV(h)}{\,\cdot\,}$, with $h\in\sectn{M}{\bfV}$.
When evaluated on any on-shell configuration $s\in\sect{M}{\bfA}$, i.e.~$P(s)=0$,
we find
\begin{flalign}
F_\varphi(s) = \int_M\vol~\ip{P_\bfV(h)}{s}= \int_M\vol~\ip{h}{P_\bfV(s)} = -\int_M\vol~\ip{h}{J} ~.
\end{flalign}
Thus, varying $h$ over $\sectn{M}{\bfV}$ those observables determine completely the source term $J$.
The same holds true for quantum observables evaluated in the class of states studied in Section \ref{sec:states},
where $J$ can be extracted from the one-point distribution.
 
To sum up, we can say that interpreting inhomogeneous linear matter field theories
as affine matter field theories naturally leads to an appropriate algebra of observables,
avoiding any choice of background solution $\widehat{s}$ and being sufficiently rich
to contain observables measuring the external source term.

%%%%%%%%%%%%%%%%%%%%%%%%%%%%%%%%%%%%%%%%%%%%%%%%
%%%%%%%%%%%%%%%%%%%%%%%%%%%%%%%%%%%%%%%%%%%%%%%%
%%%%%%%%%%%%%%%%%%%%%%%%%%%%%%%%%%%%%%%%%%%%%%%%
%%%%%%%%%%%%%%%%%%%%%%%%%%%%%%%%%%%%%%%%%%%%%%%%
\appendix

\section*{Acknowledgements}
We would like to thank Hanno Gottschalk, Thomas-Paul Hack as well as all organizers and participants of the workshop ``Algebraic Quantum Field Theory and Local Symmetries" (Bonn, 26-28/09/2012) for useful discussions and comments.
The work of C.D. is supported partly by the University of Pavia and partly
 by the Indam-GNFM project {``Effetti topologici e struttura della teoria di campo interagente''}. The work of M.B. is supported by a Ph.D. fellowship of the University of Pavia.

%%%%%%%%%%%%%%%%%%%%%%%%%%%%%%%%%%%%%%%%%%%%%%%%
%%%%%%%%%%%%%%%%%%%%%%%%%%%%%%%%%%%%%%%%%%%%%%%%

\section{\label{app:CCRandCAR}The $\CCR$ and $\CAR$ functor}
\paragraph*{Canonical commutation relations (CCR):}
Let $(\EE,\tau)$ be an object in $\VecBiLin^\mathrm{asym}$. We assign to any
element $v\in\EE$ a hermitian symbol $\Psi(v)$ and generate the free unital $\ast$-algebra $\mathcal{A}^\mathrm{free}$
over $\bbC$. We factor out the two-sided $\ast$-ideal $\mathcal{I}^{\mathrm{asym}}$ generated by the elements
\begin{subequations}
\begin{flalign}
&\Psi(\alpha \,v + \beta\,w) - \alpha\,\Psi(v) - \beta\,\Psi(w)~,\\
&\big[\Psi(v),\Psi(w)\big] - i\,\tau(v,w)\,\oone~,
\end{flalign}
\end{subequations}
for $\alpha,\beta\in\bbR$ and $v,w\in\EE$. We denote the resulting unital $\ast$-algebra by 
$\CCR(\EE,\tau):=\mathcal{A}^\mathrm{free}/\mathcal{I}^\mathrm{asym}$ and call it the
{\bf (bosonic) algebra of field polynomials}.

Let us define the category $\astAlg$ as follows:
An object in $\astAlg$ is a unital $\ast$-algebra $\mathcal{A}$ over $\bbC$.
A morphism between two objects $\mathcal{A}_1$ and $\mathcal{A}_2$ in $\astAlg$
is an injective unital $\ast$-algebra homomorphism, i.e.~an injective linear map
$\kappa :\mathcal{A}_1\to \mathcal{A}_2$ such that $\kappa(\oone_1) = \oone_2$,
$\kappa(a\,b) = \kappa(a)\,\kappa(b)$ and $\kappa(a)^\ast  = \kappa(a^\ast)$,
for all $a,b\in\mathcal{A}_1$.

The association of the algebra above is indeed a covariant functor $\CCR:\VecBiLin^\mathrm{asym}\to \astAlg$.
It associates to any object $(\EE,\tau)$ in $\VecBiLin^\mathrm{asym}$ the object
in $\astAlg$ given by $\CCR(\EE,\tau) := \mathcal{A}^\mathrm{free}/\mathcal{I}^\mathrm{asym}$.
Furthermore, given any morphism $L: \EE_1\to\EE_2$ between two objects $(\EE_1,\tau_1)$
and $(\EE_2,\tau_2)$ in $\VecBiLin^\mathrm{asym}$, we can induce an injective
unital $\ast$-algebra homomorphism $\CCR(L): \CCR(\EE_1,\tau_1)\to\CCR(\EE_2,\tau_2)$
by defining on the generating symbols $\CCR(L)(\Psi_1(v)) := \Psi_2(Lv)$, for all $v\in\EE_1$.

\paragraph*{Canonical anti-commutation relations (CAR):} Let $(\EE,\tau)$ be an object 
in $\VecBiLin^\mathrm{sym}$. We assign to any
element $v\in\EE$ a hermitian symbol $\Psi(v)$ and generate the free unital $\ast$-algebra $\mathcal{A}^\mathrm{free}$
over $\bbC$. We factor out the two-sided $\ast$-ideal $\mathcal{I}^\mathrm{sym}$ generated by the elements
\begin{subequations}
\begin{flalign}
&\Psi(\alpha \,v + \beta\,w) - \alpha\,\Psi(v) - \beta\,\Psi(w)~,\\
&\big\{\Psi(v),\Psi(w)\big\} -\tau(v,w)\,\oone~,
\end{flalign}
\end{subequations}
for $\alpha,\beta\in\bbR$ and $v,w\in\EE$. We denote the resulting unital $\ast$-algebra by 
$\CAR(\EE,\tau):=\mathcal{A}^\mathrm{free}/\mathcal{I}^\mathrm{sym}$.
We can endow the algebra $\CAR(\EE,\tau)$ with a $\mathbb{Z}_2$-grading by defining all generators
$\Psi(v)$, $v\in\EE$, to be {\it odd} and the unit $\oone$ to be {\it even}.
We call this algebra the {\bf (fermionic) algebra of field polynomials}.

We define a covariant functor $\CAR: \VecBiLin^\mathrm{sym} \to \astAlg$.
To any object $(\EE,\tau)$ in $\VecBiLin^\mathrm{sym}$ we associate the
object $\CAR(\EE,\tau) = \mathcal{A}^\mathrm{free}/\mathcal{I}^\mathrm{sym}$ in $\astAlg$.
Furthermore, given any morphism $L:\EE_1\to\EE_2$ between two objects
$(\EE_1,\tau_1)$ and $(\EE_2,\tau_2)$  in $\VecBiLin^\mathrm{sym}$, we can induce an
injective unital $\ast$-algebra homomorphism $\CAR(L): \CAR(\EE_1,\tau_1)\to\CAR(\EE_2,\tau_2)$
by defining on the generating symbols $\CAR(L)(\Psi_1(v)) := \Psi_2(Lv)$, for all $v\in\EE_1$.
This map preserves the $\mathbb{Z}_2$-grading on the algebras.

\begin{rem}
Let $(\EE,\tau)$ be an object in $\VecBiLin^\mathrm{sym}$
and consider the algebra $\CAR(\EE,\tau)$. 
Let us also assume that we have a state on this algebra, i.e.~a linear
map $\Omega: \CAR(\EE,\tau)\to \bbC$ such that
$\Omega(a)^\ast = \Omega(a^\ast)$, $\Omega(\oone) =1$ and $\Omega(a^\ast\,a) \geq 0$, for all
$a\in\CAR(\EE,\tau)$.
Consider as a special element the (equivalence class) of a generating element $\Psi(v)$,
$v\in\EE$. Due to the anti-commutation relation and hermiticity we obtain
$\Psi(v)^\ast\,\Psi(v) = \Psi(v)\,\Psi(v) = \frac{1}{2} \big\{\Psi(v),\Psi(v)\big\} = \frac{1}{2}\tau(v,v)\,\oone$.
Thus, acting with the state $\Omega$ we obtain the consistency condition
\begin{flalign}
\frac{1}{2}\tau(v,v) =  \Omega\left(\frac{1}{2}\tau(v,v)\,\oone\right) = \Omega\Big(\Psi(v)^\ast\,\Psi(v)\Big) \geq 0~.
\end{flalign}
This shows that for the existence of states it is necessary that $\tau$ is positive semi-definite.
\end{rem}
%%%%%%%%%%%%%%%%%%%%%%%%%%%%%%%%%%%%%%%%%%%%%%%%
%%%%%%%%%%%%%%%%%%%%%%%%%%%%%%%%%%%%%%%%%%%%%%%%

\end{document}